\documentclass[12pt]{article}
\usepackage{bbm}
\usepackage{t1enc}
\usepackage[latin2]{inputenc}
\usepackage{amsmath,amssymb,amsthm,upgreek}
\usepackage{txfonts}
\usepackage{sectsty}
\usepackage{overcite}
\usepackage{mathrsfs}
\usepackage{fixmath}
\usepackage[margin=1in]{geometry}
\usepackage{enumerate}
\usepackage{titlesec}
\newcommand{\midint}{\mbox{$\int$}}

\newcommand{\eq}[1]{\begin{equation}#1\end{equation}}
\newcommand{\affiliation}[1]{
              \noindent\hspace*{2.5em}\parbox[t]{15cm}{\small\it#1}
              \vspace*{2ex}\\}
\newcommand{\im}{\mathrm{i}}
\newcommand{\f}{\Upphi}
\newcommand{\trans}{M}
\newcommand{\transd}{M^*}
\newcommand{\mink}{M}
\newcommand{\minkd}{M^*}
\newcommand{\charf}{\Uptheta}
\newcommand{\vac}{\Upomega}
\newcommand{\vf}{\varphi}
\newcommand{\bs}[1]{\boldsymbol{#1}}
\newcommand{\sfr}[2]{\mbox{$\frac{#1}{#2}$}}
\newcommand{\ve}{\varepsilon}
\newcommand{\hil}{\mathcal{H}}
\newcommand{\un}{U}
\newcommand{\sch}{\mathcal{S}}
\newcommand{\supp}{\mathop{\mathrm{supp}}}
\newcommand{\smallsum}{\mathop{\mbox{$\sum$}}}
\newcommand{\Fd}[2]{\Upphi^{(#1)}{}_{\hspace{-0.6em}#2\hspace{0.4em}}}
\newcommand{\Mx}[3]{#1^{(#2)}{}_{\hspace{-0.7em}#3\hspace{0.2em}}}

\theoremstyle{plain}\newtheorem{theor}{Theorem}
\theoremstyle{plain}\newtheorem*{theor*}{Theorem}
\theoremstyle{remark}\newtheorem{prop}{Proposition} 
\theoremstyle{plain} 
\theoremstyle{remark}
\theoremstyle{remark}\newtheorem*{cor*}{Corollary}
\theoremstyle{remark}\newtheorem{lem}{Lemma}
\theoremstyle{remark}
\theoremstyle{definition}\newtheorem{defin}{Definition}

\renewcommand{\leq}{\leqslant}
\renewcommand{\geq}{\geqslant}
\titlelabel{\thetitle.\hspace{8pt}}

\renewcommand{\title}[1]{{\Large\bf\flushleft{#1}}\vspace*{3ex}\\}
\renewcommand{\author}[2]{{\noindent\hspace*{2.5em}\large#1}
                     \footnote{Electronic mail: $\mathtt{#2}$}\\}
\renewcommand{\d}{\mathrm{d}}
\renewcommand{\ddot}[1]{\mathbf{\dot{\mathit{#1}}}}
\sectionfont{\normalsize\bf}
\subsectionfont{\normalsize\bf}
\makeatletter\renewcommand\@biblabel[1]{$^{#1}$}\makeatother
\begin{document}
\title{On scale symmetry in gauge theories}
\author{Szilard Farkas}{farkas@uchicago.edu}

\affiliation{Enrico Fermi Institute and Department of Physics, University of Chicago,\\ Chicago, IL 60637, USA}
\vspace{-0ex}

\begin{abstract}
Buchholz and Fredenhagen proved that particles in the vacuum sector of a scale invariant local quantum field theory do not scatter. More recently, Weinberg argued that conformal primary fields satisfy the wave equation if they have nonvanishing matrix elements between the vacuum and one-particle states. These results do not say anything about actual candidates for scale invariant models, which are nonconfining Yang-Mills theories with no one-particle states in their vacuum sector. The BRST quantization of gauge theories is based on a state space with an indefinite inner product, and the above-mentioned results do not apply to such models. However, we prove that, under some assumptions, the unobservable basic fields of a scale invariant Yang-Mills theory also satisfy the wave equation. In ordinary field theories, particles associated with such a dilation covariant hermitian scalar field do not interact with each other. In the BRST quantization of gauge theories, there is no such triviality result.   
\end{abstract}

\section*{INTRODUCTION}
The more symmetries a quantum field theory has, the more we can say about its properties without explicit constructions. By exploring the consequences of symmetries, we might gain enough insight into the structure of the model so that its actual construction becomes feasible. Or we may discover that the assumed symmetries bring about too much simplification, and only trivial models can respect them. Buchholz and Fredenhagen's no-go theorem\cite{BuchholzFred} about scale invariant field theories belongs to the second category. They showed in a general setting of relativistic local quantum field theory that if the Poincar\'e representation extends to dilations and there are one-particle states in the vacuum sector so that their collisions can be defined, the only scattering matrix invariant under dilations is the identity operator. 

On the other hand, $\mathcal{N}=4$ supersymmetric Yang-Mills (SYM) theory is believed to exist both as a local quantum field theory and as a scattering theory. Fortunately, its scale invariance does not bode ill with this prospect. The theory escapes the implications of the Buchholz and Fredenhagen's no-go theorem by a pathology. Its off-shell scattering amplitudes are plagued by infrared divergences, so the theory does not predict the kind of scattering processes assumed by the no-go theorem. However, we prove a theorem that may belong to the first category, that is, it shows that the BRST quantization might be a promising framework for the construction of $\mathcal{N}=4$ SYM theory. First we give a more detailed account of Buchholz and Fredenhagen's no-go theorem and motivate the assumptions that lead to it.

In the traditional framework of field theory, the basic object is a field $\f$ that transforms covariantly under a positive energy continuous unitary representation $\un$ of the covering group of the Poincar\'e group on a separable Hilbert space $\hil$ on which the field operators are defined:
\[\un(a,L)\,\f_i(x)\,\un(a,L)^*=\smallsum_j S_{ij}(L^{-1})\,\f_j(Lx+a),\]
where $\un(a,L)$ is the operator representing a Lorentz transformation $L$ followed by a translation $a$. Let $\mathcal{O}$ be a bounded open region of the Minkowski spacetime and define $\mathfrak{F}(\mathcal{O})$ as the algebra generated by the spectral projections of all the self-adjoint smeared field operators $\f(f)$, where $f$ is supported within $\mathcal{O}$. Then $\un(L,a)\,\mathfrak{F}(\mathcal{O})\,\un(L,a)^*=\mathfrak{F}(L\mathcal{O}+a)$. Furthermore, it follows from the locality of fields that $[\mathfrak{F}(\mathcal{O}_1),\mathfrak{F}(\mathcal{O}_2)]=0$ if $\mathcal{O}_1$ and $\mathcal{O}_2$ are spacelike separated. A more general setting of field theory starts with a net $\mathcal{O}\mapsto\mathfrak{F}(\mathcal{O})$ of algebras of bounded operators on $\hil$ which has these properties and a $\un$-invariant unit vector $\vac\in\hil$, the vacuum, which is unique up to a phase. The elements of the local algebras $\mathfrak{F}(\mathcal{O})$ are called local operators. Here we considered only bosonic operators, but fermionic operators can be incorporated by straightforward modifications. The subsequent considerations apply to this more general case as well.

Suppose that there are massless one-particle states in the vacuum sector. This means that there is a family of subspaces $\hil_i$ in $\hil_0=\overline{\{\,A\vac\,|\,A\in\mathfrak{F}(\mathcal{O})\,\}}$ on which $\un$ is an irreducible massless representation of some helicity $s_i$. Let $\mathcal{F}$ be the Fock space constructed by the usual procedure based on one-particle states that furnish a Poincar\'e representation equivalent to the restriction of $\un$ to the direct sum of $\hil_i$. The Poincar\'e representation obtained by this construction on $\mathcal{F}$ will be denoted by $\un_0$. The structure introduced in the previous paragraph allows us to construct two continuous linear isometric embeddings $W^{\mathrm{in}}$ and $W^{\mathrm{out}}$ of the Fock space $\mathcal{F}$ into $\hil_0$ which are Poincar\'e invariant: $\un(L,a)\,W^{\mathrm{in/out}}=W^{\mathrm{in/out}}\,\un_0(L,a)$. The one-particle representations can be extended to representations of dilations, which in turn give rise to a unitary representation $\un_0(\lambda)$ of dilations on $\mathcal{F}$. So $\un_0(\lambda)$ is a continuous unitary representation of the multiplicative group of the positive real numbers satisfying $\un_0(\lambda)\,\un_0(a,L)=\un_0(\lambda a,L)\,\un_0(\lambda)$. Now we are in the position to quote the no-go theorem. Suppose that $\hil^\mathrm{ext}\coloneqq W^\mathrm{in}\,\mathcal{F}=W^{\mathrm{out}}\,\mathcal{F}$ and the operators $\un^\mathrm{in/out}(\lambda)$ defined on $\hil^\mathrm{ext}$ by the relation $\un^\mathrm{in/out}(\lambda)\,W^\mathrm{in/out}=W^\mathrm{in/out}\,\un_0(\lambda)$ can be extended to the same unitary operator $\un(\lambda)$ on $\hil_0$ so that $\lambda\mapsto\un(\lambda)$ is a continuous representation of dilations that satisfies the same multiplication rule with $\un(a,L)$ as $\un_0(\lambda)$ with $\un_0(a,L)$. Then $W^\mathrm{in}=W^\mathrm{out}$. In particular, if the theory is asymptotically complete in the vacuum sector, that is, if $\hil^\mathrm{ext}=\hil_0$, then the $S$-matrix is trivial: $S=(W^\mathrm{out})^{-1}W^\mathrm{in}=\mathbbm{1}$.

Since scale invariance precludes collisions in the vacuum sector, the only way to construct dilation covariant field theoretic models for scattering processes is to incorporate scattering states that are not in the vacuum sector. In other words, we need scattering states which are not the elements of the Hilbert space generated by $\mathfrak{F}(\mathcal{O})\vac$. This is precisely how $\mathcal{N}=4$ SYM theory gets around the no-go theorem. The model is not in the scope of validity of the no-go theorem because the asymptotic one-particle states, if they exist, have nontrivial quantum numbers associated with charges that are related to a gauge field strength by a local Gauss law, and therefore they cannot be in the vacuum sector.

In order to avoid complications originating in the lack of gauge invariance of the field strength of a non-Abelian gauge theory, we shall recapitulate the implications of Gauss's law for electromagnetism. The Lagrangian  
\[
\mathscr{L}=\bar{\psi}(\mathrm{i}\gamma^\mu\partial_\mu-m)\psi-\mbox{$\frac{1}{4}$}F^{\mu\nu}F_{\mu\nu}-J^\mu A_\mu,\;\;\;\;\;F_{\mu\nu}=\partial_\mu A_\nu-\partial_\nu A_\mu,\;\;\;\;\;J^\mu=e\bar{\psi}\gamma^\mu\psi,
\]
is invariant under the gauge transformation $A_\mu\to A_\mu+\partial_\mu\alpha$, $\psi\to e^{-\mathrm{i}e\alpha}\psi$, where $\alpha$ is a smooth function on the spacetime. The corresponding Noether charge is 
\[
Q_\alpha=\midint\d^3\!x\left[\mbox{$\frac{\partial\mathscr{L}}{\partial\partial_0 A_\mu}$}\partial_\mu\alpha+\mbox{$\frac{\partial\mathscr{L}}{\partial\partial_0\psi}$}(-\mathrm{i}e\alpha\psi)\right]=\midint\d^3\!x\left[\bs{E\cdot\nabla}\alpha+J^0\alpha\right],\;\;\;E^i=F^{0i},
\]
which is defined only if $\alpha$ has appropriate asymptotic behavior. If it is compactly supported, then integration by parts shows that $Q_\alpha=0$ if the equation of motion is satisfied, in accordance with the general theorem about the vanishing of the Noether charge associated with a local symmetry.\cite{Lee} $Q_\alpha$ are the constraints that generate the local symmetries in the canonical formalism. In quantum theory, their vanishing on the physical Hilbert space corresponds to the gauge invariance of observables. Since $Q_\alpha=0$ implies Gauss's law, the electric field in quantum electrodynamics is expected to satisfy the equation $\bs{\nabla}\bs{\cdot}\bs{E}=J^0$, where $J^0$ is a field in terms of which the electric charge is defined as follows. Classically, the electric charge is the Noether charge of symmetry transformations parametrized by a constant $\alpha$, which is $Q_1=\smallint\d^3\!x\,J^0$. In quantum theory, $J^0$ is a field that should be smeared with compactly supported functions on the spacetime, so the operator corresponding to the classical expression of $Q_1$ is defined by a limit        
\begin{align*}
[\mathrm{Q}_1,A]=\lim_{r\to\infty}[\bs{\nabla}\!\bs{\cdot}\!\bs{E}(\varphi_r),A]=\lim_{r\to\infty}[-E^i(\partial_i\varphi_r),A],\;\;\;&\varphi_r(t,\mathbf{x})=\tau(t)f\left(\mbox{$\frac{|\mathbf{x}|}{r}$}\right),\;\;\;\smallint\d t\,\tau(t)=1,\\&f(\xi)=1\;\;\mbox{if}\;\xi\leq 1,\;\;\;\;f(\xi)=0\;\;\mbox{if}\;\xi\geq 2,\nonumber
\end{align*}
where $A$ is an operator, $f$ and $\tau$  are smooth functions. We ignore technicalities and do not specify in what sense the limit is taken and for what $A$ it should exist. Since the support of $\partial_i\varphi_r$ is spacelike separated from any bounded set for large enough $r$, we get $[\mathrm{Q}_1,A]=0$ if $A$ is local relative to the electric field, that is, if there is a bounded region $\mathcal{O}$ such that $A$ commutes with the electric field smeared with a test function supported in the causal complement of $\mathcal{O}$. So we conclude that local operators are neutral, and since $Q_1\vac=0$, charged states cannot be created from the vacuum by local operators.

As in electromagnetism, the Noether charges $Q_\alpha=\smallint\d^3\!x\sum_a\alpha_a(J^{a0}-\bs{\nabla}\bs{\cdot}\bs{E}^a)$ associated with the gauge symmetries are their generators in the canonical formalism of classical Yang-Mills theories. Here $a$ labels the components in a basis of the Lie-algebra of the gauge group. Unlike in electromagnetism, where Gauss's law is a relation between gauge invariant quantities, in Yang-Mills theories the implementation of Gauss's law as a field equation requires the introduction of fields that are not gauge invariant. Such a departure from physical quantities may be useful at some intermediate steps of a construction. One example is the indefinite metric formulation of Yang-Mills theories, also known as BRST formalism, where the physicality condition on states is imposed as the requirement that for any compactly supported $\alpha$, equation $\big(\vf,\sum_a(J^{a0}-\bs{\nabla}\bs{\cdot}\bs{E}^a)(\alpha_a)\psi\big)=0$ holds for the elements $\vf$ and $\psi$ of a subspace $\hil'$ on which the inner product $(\cdot\,,\cdot)$ is semidefinite. Let the vacuum $\vac$ be physical and neutral: $\vac\in\hil'$ and $Q_1^a\vac=0$, where $Q_1^a=\smallint\d^3\!x\,J^{a0}$ are the Noether charges corresponding to global symmetries. These charges are defined similarly to the electric charge in quantum electrodynamics. If $A$ is an operator such that $A\vac$ is in $\hil'$ and the charges $Q_1^a$ generate an observable infinitesimal transformation of $A\vac$: $(\vf,Q_1^aA\vac)\neq0$ for some $\vf\in\hil'$, then $A$ cannot be local relative to $\bs{E}^a$. We shall consider field theories in an indefinite metric because this is the structure underlying the BRST quantization. The BRST formalism allows for charged local operators, and locality played an important role in the no-go theorem, so it is this formulation where analogous results may be expected. 

\subsection*{{\it Organization of the paper}}

Section \ref{GaugeLocal} is an introduction to the indefinite metric formulation, and more specifically, to the BRST quantization. Section \ref{Wightman} provides a summary of the generalized Wightman axioms applicable to this formalism. This section starts with our conventions. After some preparations in Section \ref{Representation} and \ref{General}, the proofs in Section \ref{Final} are free of technicalities. The main result is the Corollary. So a quick reading of the paper may start with the Corollary in Section \ref{Final} and continue with Section \ref{Discussion}.

The Corollary states that if the metric operator that specifies the inner product commutes with translations, then any local covariant field in a scale invariant theory satisfies the wave equation if it has nonvanishing matrix elements between the vacuum and physical one-particle states (one-particle states of ``nonzero norm'') or if its scaling dimension is small enough. The theorem quoted in Section \ref{Reeh-schlieder} shows that the commutator of such a field is a $c$-number.

In section \ref{Discussion}, we propose a plausibility argument showing that the conditions of the Corollary are met by the basic physical (not ghost) fields of $\mathcal{N}=4$ SYM theory. Therefore they satisfy the wave equation and their commutator is a $c$-number. Basic fields are the fields corresponding to the variables of the classical action. The importance of the condition that these fields have nonvanishing matrix elements between the vacuum and physical one-particle states is that it puts an upper bound on their scaling dimensions. If this condition is not satisfied, it may be problematic to develop a scattering theory. Less speculative is the observation that the scaling dimensions inferred from the properties of the so-called protected (BPS) operators respect the above-mentioned upper bounds. This is another way to justify the conditions of the Corollary. In Section \ref{Discussion} we also discuss the assumptions made in addition to the generalized Wightman axioms.

\section{GAUGE SYMMETRIES IN LOCAL FIELD THEORIES}\label{GaugeLocal}

We are interested in the BRST quantization of field theories only as a technique, a structure that might facilitate some constructions. To indicate that this formalism has promising prospects, all that we can do at this point is to demonstrate that it elegantly encapsulates the principles of gauge theories. This can be done succinctly in the canonical formalism, with an emphasis on local symmetries and the phase space constraints they imply. Our presentation is rather sketchy because the arguments of the paper do not rely on the details of the construction. The only purpose of this section is to provide support for this quantization technique.

For simplicity, we shall consider fields defined on the Minkowski spacetime $\mink$ equipped with a standard synchronization. The canonically conjugate variables of a Yang-Mills field can be chosen to be the spatial component $\bs{A}$ of the vector potential and the electric field $\bs{E}$ whose components in terms of the field strength are $E^i=F^{0i}$. (Lie algebra indices are suppressed.) The construction of the phase space\cite{Lee} is based on a specific choice of a Cauchy surface $\Sigma$. In the present case, the canonical variables at different times can be mapped into each other by time translations. This allows us to describe the collection of phase space variables of a field configuration on spacetime at different times as a trajectory in a single phase space parametrized by the time. In addition to the canonical fields denoted schematically by $Q$ and $P$, which include the variables that fully specify the matter configuration, there are noncanonical fields $N$. The noncanonical component of a Yang-Mills field is time component of the vector potential.     

The variation of the action with respect to $N$ yields the constraints 
\[
\mathscr{H}[Q,P]=0,
\]
which are implied by the local symmetries of the Lagrangian.\cite{Lee} The Poisson algebra of the phase space functionals $(P,Q)\mapsto\smallint{}_{\!\!\Sigma\;}\lambda\cdot\mathscr{H}[P,Q]$ is closed. They generate the gauge transformations of the canonical fields. We introduced the notation $\cdot$ for summation over the (suppressed) index labeling the symmetry generators. The field $N$ can be turned into a canonical field by introducing a canonically conjugate momentum $R$ and imposing the constraint $R=0$. The fields $Q,P,N$, and $R$ will be denoted collectively by $\Upphi$. Let $\tilde{\mathscr{H}}$ be the expression obtained from $\mathscr{H}$ by adding a term bilinear in $N$ and $R$ such that the functionals $\Upphi\mapsto\smallint{}_{\!\!\Sigma\;}\lambda\cdot\tilde{\mathscr{H}}[\Upphi]$ generate the gauge transformations of all the field variables, including $N$ and $R$. 

Let $J$ be the Noether current of the global symmetry associated with the gauge symmetries. If $\Upphi$ is a solution, then 
\eq{\label{ConstraintsOriginal}
\partial_\mu F^{\mu\nu}+J^\nu=0,\;\;\;\;\;\;\;R=0,
}
and since $\partial_iF^{i0}+J^0=\mathscr{H}$, we can also write 
\eq{\label{generators_fieldeq}\partial_iF^{i0}+J^0=\tilde{\mathscr{H}}.} 
 
The BRST quantization of Yang-Mills theories starts with a modification of the classical action. A gauge fixing term is added to the Lagrangian, which eliminates the local symmetries and enlarges the phase space. For example, take the term 
\eq{\label{gauge_fixing}
\mathscr{L}_\mathrm{gf}=R\cdot\partial_\mu A^\mu+\sfr{\alpha}{2}R\cdot R,
}
where $\alpha$ is a nonzero constant, and $R$ is the Nakanishi-Lautrup auxiliary field. Then $N$ and $R$ are canonical fields from the beginning and we do not have to append them somewhat artificially to the phase space variables as before. A ghost term is also added to the Lagrangian, which can be motivated by the appearance of the same term in the path integral, and it leads to a simple characterization of local observables and physical states. Now $Q$ and $P$ include the phase space variables of the ghost and antighost fields in addition to those of the gauge and matter fields.  Inspection of the field equations in [\citeonline{Kugo}] tells us that Eq.~\eqref{generators_fieldeq} is still valid, with $\smallint{}_{\!\!\Sigma\;}\lambda\cdot\mathscr{H}$ containing two terms, the original one and a ghost term. But instead of Eq.~\eqref{ConstraintsOriginal}, we have   
\eq{\label{ConstraintsFixed}
\partial_\mu F^{\mu\nu}+J^\nu=\{\xi^\nu,\mathrm{Q}\},\;\;\;\;\;R=\{\eta,\mathrm{Q}\}.
}
The particular form of the vector field $\xi$, the scalar field $\eta$, and the conserved BRST charge $\mathrm{Q}$ is not important for the present discussion. At any time, they can be expressed in terms of the instantaneous spatial configuration of the fields. Their Poisson brackets are defined by the aid of these expressions. Since the ghost fields are anticommuting, by the usual construction of the phase space, we find that the Poisson brackets on the right hand side are symmetric. They are replaced by anticommutators in quantum theory.

Let $F$ be a functional on the phase space. In order for $F$ to be considered a symmetry generator, the simplest possibility is that $F$ be the element of a finite dimensional subspace of phase space functionals which contains the Hamiltonian and forms a Lie-algebra with the Poisson bracket as the Lie-bracket. If we are to include local symmetries, instead of a finite dimensional subspace, we may consider the space of functionals of the form $\Upphi\mapsto\smallint{}_{\!\!\Sigma\;}\lambda\cdot\mathscr{F}[\Upphi]$, where  $\lambda$ is a compactly supported smooth function on $\Sigma$ which takes its value in a finite dimensional vector space. Then the requirement is that these generators and the Hamiltonian form a closed Poisson algebra. In some cases, it may be necessary to allow for $\Upphi$-dependent $\lambda$. The gauge fixing term in Eq.~\eqref{gauge_fixing} is invariant under a gauge transformation parametrized by a function $\Lambda$ that satisfies $\partial_\mu D^\mu\Lambda=0$. In electrodynamics, these symmetries are sometimes called residual gauge symmetries. They are generated by the functionals 
\eq{\label{generators}\Upphi\mapsto\smallint{}_{\!\!\Sigma\;}\lambda\cdot\tilde{\mathscr{H}}[\Upphi],\;\;\;\;\;\;\;\Upphi\mapsto\smallint{}_{\!\!\Sigma\;}\lambda\cdot R,} 
whose Poisson algebra with the Hamiltonian is closed.

In Yang-Mills theories, the gauge fixing term in Eq.~\eqref{gauge_fixing} leads to a ghost term by the Faddeev-Popov prescription\cite{Weinberg2} with which the functionals in Eq.~\eqref{generators} do not form a closed Poisson algebra with the Hamiltonian, so unlike in electromagnetism, residual gauge symmetries are not implemented on the phase space. However, the operators corresponding to these functionals still have physical significance in quantum field theory, where physical states are selected by the aid of the BRST charge so that the matrix elements of the anticommutators on the right hand side of Eq.~\eqref{ConstraintsFixed} vanish between them. So in this sense, Eq.~\eqref{ConstraintsOriginal} is regained on the physical state space.

Now we elaborate a little further on the selection of physical states and the role of the BRST charge in quantum field theory. Field operators are defined on a Hilbert space $\hil$, in which an additional inner product $(\cdot\,,\cdot)$ is introduced. The inner product $(\cdot\,,\cdot)$ is nondegenerate, as a consequence of a compatibility requirement discussed later, which makes the structure amenable to techniques in functional analysis. The construction of observables and physical states is based on a closed subspace $\hil'$ on which the inner product $(\cdot\,,\cdot)$ is positive semidefinite. Let $A^+$ be the adjoint of an operator $A$ with respect to $(\cdot\,,\cdot)$. The set $\mathscr{A}$ of field operators is required to have the following properties. There is a dense subspace $D$ of $\hil$ such that $D$ is contained by the domain of each element of $\mathscr{A}$ and $\mathscr{A}D\subset D$. The vacuum $\vac$ is in $D$. Furthermore, $\mathscr{A}$ is an algebra, which is closed under involution: $A^+\in\mathscr{A}$ if $A\in\mathscr{A}$.

The time dependent operators corresponding to the classical expressions in Eq.~\eqref{generators} will be denoted by $\tilde{\mathscr{H}}(\lambda)$ and $R(\lambda)$, respectively. It is assumed that they are defined on $D$, which they leave invariant. It may be unnecessarily restrictive from a physical point of view to define gauge invariance of $A\in\mathscr{A}$ by the condition that $[A,\tilde{\mathscr{H}}(\lambda)]=0$ and $[A,R(\lambda)]=0$ hold on $D$. Instead, it is enough if for any $\lambda$,
\eq{\label{GaugeInvOp}\big(\vf,[A,\tilde{\mathscr{H}}(\lambda)]\psi\big)=0,\;\;\;\;\big(\vf,[A,R(\lambda)]\psi\big)=0\;\;\;\;\;\mbox{for all}\;\;\vf,\psi\in D'=\hil'\cap D.}  
Motivated by the fact that gauge symmetries are trivially represented on the physical Hilbert space, we also require that 
\eq{\label{GaugeInvState}\big(\vf,\tilde{\mathscr{H}}(\lambda)\psi\big)=0,\;\;\;\;\big(\vf,R(\lambda)\psi\big)=0\;\;\;\;\;\mbox{for all}\;\;\vf,\psi\in D'.}
With this requirement, condition \eqref{GaugeInvOp} follows from the property 
\eq{\label{GaugeInvBRST}
AD'\subset D',\;\;\;\;A^+D'\subset D'.
}
Gauge invariance in the BRST quantization is defined by condition \eqref{GaugeInvBRST}.
 
Ideally, the enlargement of the physical state space $\hil'$ to $\hil$ is the necessary minimum, that is, $\hil'$ is maximal in the sense that the set $\hil\setminus\hil'$ does not contain vectors orthogonal to $\hil'$ with respect to $(\cdot\,,\cdot)$. Clearly, with a positive definite inner product, this condition would imply that $\hil=\hil'$. Then the Noether charge associated with the global symmetry would generate a trivial transformation on local operators. Since one of the motivations for the BRST quantization is the possibility of charged local operators, we would like to avoid this implication. Therefore $\hil'$ can be maximal only in an indefinite metric. Condition \eqref{GaugeInvState} can be achieved by virtue of \eqref{ConstraintsFixed} if we choose $\hil'=\mathop{\mathrm{ker}}\mathrm{Q}$, where $\mathrm{Q}=\mathrm{Q}^+$. For simplicity, it is assumed that $\mathrm{Q}$ is everywhere defined and therefore bounded since a pseudo-Hermitian operator is closed. For any operator $A$ on $\hil$ whose adjoint with respect to the Hilbert space scalar product is densely defined, $\overline{\mathop{\mathrm{im}}A}=(\mathop{\mathrm{ker}}A^+){}^\perp$, where the superscripts $+$ and $\perp$ indicate the adjoint operator and the orthogonal complement with respect to $(\cdot\,,\cdot)$. The completion of $\mathop{\mathrm{im}}A$ is meant in the topology induced by the Hilbert space norm. Since $\mathrm{Q}=\mathrm{Q}^+$, maximality of $\hil'$ is equivalent to $\mathrm{Q}^2=0$. The inner product $(\cdot\,,\cdot)$ vanishes on $\hil''=\overline{\mathop{\mathrm{im}}\mathrm{Q}}$. If $A$ is gauge invariant, that is, it satisfies condition \eqref{GaugeInvBRST}, then $AD''\subset D''$, where $D''=\hil''\cap D$, and $A$ gives rise to a densely defined operator $A^\mathrm{ph}$ on the physical Hilbert space $\hil^\mathrm{ph}=\overline{\hil'/\hil''}$ defined by the usual completion procedure.

The indefinite metric formulation also makes it possible to define a local vector potential operator. On the physical Hilbert space, this is impossible. If $A$ is a covariant operator that transforms in the $(j,k)$ representation of the Lorentz group and has nonvanishing matrix elements between the vacuum and helicity $h$ states, then $|j-k|=|h|$.\cite{WeinbergMassless} In the free field theory of pure radiation field, the Faraday tensor creates helicity $\pm1$ particles from the vacuum, and in electrodynamics, it is still expected to have nonzero matrix elements between these photon states and the vacuum. Then $A$ also has nonzero matrix elements between photon states and the vacuum, but it is supposed to be a vector field, so $j=k=\frac{1}{2}$, which violates the condition $|j-k|=1$. That is why the vector potential operator in any direct construction on the physical Hilbert space, such as quantum electrodynamics in the Coulomb gauge, transforms under the unitary representation $U(\Lambda)$ of a boost $\Lambda$ as a vector field only up to a term resembling a gauge transformation:\cite{Weinberg1}
\[
U(\Lambda)A_\mu(x)U(\Lambda)^*=\smallsum_\nu\Lambda^\nu{}_\mu A_\nu(\Lambda x)+\partial_\mu\Omega(x,\Lambda).
\]
In the indefinite metric formulation, there is no such kinematical constraint on $A$. Indeed, in the Gupta-Bleuler quantization of pure radiation field, there is no difficulty defining a local covariant vector potential $A$ whose antisymmetric derivative $F=dA$ gives rise to a physical field $F^\mathrm{ph}$ identical with the free Faraday tensor.

\section{GENERALIZED WIGHTMAN AXIOMS}\label{Wightman}
The Lorentz product of two vectors $x$ and $y$ of the Minkowski spacetime $\mink$ will be denoted by $x\cdot y$. The metric signature is $(-,+,+,+)$. The d'Alembertian operator is $\Box=\partial_0^2-\Updelta$. The Fourier transform of a function $f:\mink\to\mathbb{C}$ and the inverse Fourier transform of $g:\minkd\to\mathbb{C}$ are defined as
\[
\hat{f}(p)=\midint_{\!\!\mink\,}\d^4\!x\,e^{\im p\cdot x}f(x),\;\;\;\;\;\;\;\check{g}(x)=\sfr{1}{(2\pi)^4}\midint_{\!\!\minkd\,}\d^4\!p\,e^{-\im p\cdot x}f(p).
\]
Complex conjugation is indicated by a bar, as in $\bar{c}$. If $A$ is an operator on a Hilbert space, its adjoint is $A^*$. The notation $A^+$ is preserved for the pseudo-adjoint, which will be defined shortly. The translation group of $\mink$ is identified with $\mink$ itself. Let $G$ be a group acting on $\trans$. The point that $g\in G$ assigns to $x\in\mink$ will be denoted by $gx$. Suppose that $gx$ is a linear function of $x$. The Lorentz product gives rise to a natural isomorphism between $\trans$ and its dual $\transd$, which allows us to define the action of $G$ on $\transd$ by $(gp)\cdot x\coloneqq p\cdot(g^{-1}x)$. The forward light cone is $V{}^+=\{\,p\in\transd\,|\,p^2<0, p^0>0\,\}$, and $\overline{V}{}^+$ is its closure. (The forward light cone in $\mink$ is defined in the same way and denoted by the same symbol.) Its boundary is $\partial V^+$. The elements of $\transd$ are sometimes referred to as momenta.

Let $L$ be an element of $S\!L(2,\mathbb{C})$, which is the universal covering group of the proper Lorentz group. The action of $L$ on $\trans$ is the Lorentz transformation assigned to $L$ by the covering map. From now on, whenever we say Lorentz or Poincar\'e transformation (or representation), we always mean an element (or a representation) of $S\!L(2,\mathbb{C})$ or $IS\!L(2,\mathbb{C})$, respectively. Consider the maps $s(\lambda)$ that assign $\lambda x$ to $x\in\mink$, where $\lambda$ is a positive real number. They form the group $\mathcal{D}$ of scale transformations or dilations. This group is defined by its action on $\trans$, so $s(\lambda)x=\lambda x$. Note that the corresponding action on $p\in\transd$ is given by $s(\lambda)p=\lambda^{-1}p$. The direct product of $\mathcal{D}$ and $S\!L(2,\mathbb{C})$ will be called the group of dilations and Lorentz transformations. The group of dilations and Poincar\'e transformations is the semidirect product of $IS\!L(2,\mathbb{C})$ and $\mathcal{D}$ with the action of $\mathcal{D}$ on $IS\!L(2,\mathbb{C})$ defined as $s(\lambda)(a,L)\coloneqq(\lambda a,L)$, where $(a,L)$ is a Lorentz transformation $L$ followed by a translation $a\in\mink$. We shall consider representations of (semi)direct products of groups specified by operators $\un(a,g,\dots)$ that represent the transformations obtained by a successive application of the group elements going from the right to the left in the argument of $\un$. When some of these are the identity element $1$ of the corresponding group, they are simply omitted. That is, we write for example $\un(a)$ instead of $\un(a,1)$.         

The definition of fields and a field theory in an indefinite metric is analogous to the Wightman axioms\cite{Streater} formulated for fields on a Hilbert space, but the inner product is allowed to be indefinite, so the state space is a pseudo-Hilbert space.  
\begin{defin}
A \emph{pseudo-Hilbert space} is a complex Hilbert space $\hil$ endowed with a possibly indefinite inner product $(\cdot\,,\cdot)$ which is compatible with the scalar product $\langle\cdot\,,\cdot\rangle$ of $\hil$, that is, for every continuous functional $F:\hil\to\mathbb{C}$ there is a unique $\varphi\in\hil$ such that $F(\psi)=(\varphi,\psi)$ for all $\psi\in\hil$.    
\end{defin}
The choice of the Hilbert space scalar product $\langle\cdot\,,\cdot\rangle$ in $\hil$ is to a certain degree inessential.\cite{Bogolubov} Its purpose is to define a topology, the one induced by the norm $\|\varphi\|=\sqrt{\langle\varphi,\varphi\rangle}$. This topology is uniquely defined by the condition of compatibility of the inner product $(\cdot\,,\cdot)$ with $\langle\cdot\,,\cdot\rangle$. Any concept relying on a topology or norm, such as separability, closure or denseness of sets, continuity or boundedness of operators, is meant with respect to the norm $\|\cdot\|$.

The pseudo-Hermitian conjugate $A^+$ of a linear operator $A$ on $\hil$ with a dense domain $D$ is defined on $\vf\in\hil$ if $D\to\mathbb{C}:\psi\mapsto(\vf,A\psi)$ is continuous, and then $A^+\vf$ is defined by $(A^+\vf,\psi)=(\vf,A\psi)$ for all $\psi\in D$. A linear bijection $\un$ is pseudo-unitary if its domain $D$ and image are dense subspaces of $\hil$ and $(\un\vf,\un\psi)=(\vf,\psi)$ for all $\vf,\psi\in D$. A pseudo-unitary representation of a topological group $G$ is a representation $g\mapsto\un(g)$ by pseudo-unitary operators $\un(g)$, with a common dense domain $D$ such that $\un(g)D=D$ for all $g\in G$. The representation is called continuous if $g\mapsto\langle\vf,\un(g)\psi\rangle$ is continuous for all $\vf,\psi\in D$. 

Let $\hil$ be a pseudo-Hilbert space with inner product $(\cdot\,,\cdot)$. Suppose that $\f:f\mapsto\f(f)$ assigns an operator $\f(f)$ to each test function $f$ in the Schwartz space $\sch(\mink)$ such that the operators $\f(f)$ have a common dense domain $D$, which is called the domain of $\f$. If $\sch(\mink)\to\mathbb{C}:f\mapsto(\vf,\f(f)\psi)$ is continuous for all $\vf,\psi\in D$, then we say that $f$ is an operator-valued distribution on $\mink$. The pseudo-adjoint $\f^+$ of $\f$ is defined by $\f^+(f)\coloneqq\f(\bar{f})^+$. Now we are ready to define the generalized Wightman axioms.
    
\begin{defin}\label{deflcqft}
Let $\hil$ be a separable pseudo-Hilbert space with inner product $(\cdot\,,\cdot)$, $\un$ a continuous pseudo-unitary Poincar\'e representation, and $\Upphi$ a finite collection of operator-valued distributions $\Fd{n}{i}$ on $M$. We say that $\Upphi$ and $\un$ define a \emph{local covariant quantum field theory} if they have the following properties:
\begin{enumerate}[I.]
\item \label{domain}\emph{Domain.} The operators $\Fd{n}{i}(f)$ have a common dense domain $D$ such that $\Fd{n}{i}(f)D\subset D$.  
\item \label{covariance}\emph{Covariance.} For all $(a,L)\in IS\!L(2,\mathbb{C})$, the pseudo-unitary operators $\un(a,L)$ are defined on $D$, $\un(a,L)D=D$, and the equality 
\[
\un(a,L)\,\Fd{n}{i}(f)\un(a,L)^{-1}=\smallsum_j\Mx{S}{n}{ij}(L^{-1})\,\Fd{n}{\!j}\big(\{a,L\}f\big)
\]
holds when both sides act on a vector in $D$. Here $S^{(n)}$ are some Lorentz representations. We introduced the notation $\big(\{a,L\}f\big)(x)\coloneqq f(L^{-1}(x-a))$.
\item \label{spectrum}\emph{Spectral Condition.} For any $\varphi,\psi\in D$, the functions $a\mapsto(\vf,\un(a)\psi)$ are required to be polynomially bounded so that the Fourier transform symbolically written as
\[
\smallint e^{-\im p\cdot a}(\vf,\un(a)\psi)
\]  
exists as a tempered distribution, whose support lies in $\overline{V}{}^+$. 
\item \label{vacuum}\emph{Vacuum.} There is an $\vac\in D$, called the \emph{vacuum}, such that $(\vac,\vac)=1$ and $\un(a,L)\vac=\vac$ for all $(a,L)$. Any other translation invariant vector differs from $\vac$ by a complex phase.  
\item \label{locality}\emph{Locality.} If the supports of the test functions $f$ and $g$ are spacelike separated, that is, if $(x-y)^2>0$ for all $x\in\supp f$ and $y\in\supp g$, then 
\[
[\Fd{n}{i}(f),\Fd{m}{\!j}(g)]_\mp=0
\]
on $D$, where, depending on $n$ and $m$, $[\cdot,\cdot]_\mp$ is either a commutator or anticommutator.
\item \label{conjugation}\emph{Conjugation.} For any $n$ there is an $\tilde{n}$ such that $\Fd{\tilde{n}}{i}=(\Fd{n}{i}){}^+$.
\item \label{cyclicity}\emph{Cyclicity.} The linear space generated by vectors of the form $\Fd{n_1}{\!i_1}\!(f_1)\dots\Fd{n_\ell}{\!i_\ell}\!(f_\ell)\vac$ is dense in $\hil$. 
\end{enumerate}
\end{defin}
The last axiom can be written as $\overline{\mathcal{P}\vac}=\hil$, where $\mathcal{P}$ is the polynomial field algebra. This algebra consists of the sums of constant multiples of the identity operator and operators of the form $\Fd{n_1}{\!i_1}\!(f_1)\dots\Fd{n_\ell}{\!i_\ell}\!(f_\ell)$, where $f_1,\dots,f_\ell\in\sch(\mink)$. Let $O$ be an open subset of $\mink$. The polynomial field algebra $\mathcal{P}(O)$ of $O$ is defined in the same way as $\mathcal{P}$, except that the test functions are all supported in $O$.
  
If $(\cdot\,,\cdot)$ is positive definite, then these axioms reduce to the ordinary Wightman axioms. The index $n$ labels particle species. An operator-valued distribution on $\mink$ will be called a field. When there is a pseudo-unitary representation of some topological group $G$ which satisfies conditions analogous to Axiom \ref{covariance}, we say that a field is covariant under $G$. If $G$ is unspecified, covariance means covariance under Poincar\'e transformations. If $S$ is the representation that appears in the transformation rule \ref{covariance} of a covariant field $\f$, we say that $\f$ transforms in $S$.

Distributions are difficult to analyze directly. The following theorem,\cite{Streater} which is a consequence of the positivity of energy, sometimes makes it possible to study the properties of distributions occurring in field theory by the analysis of holomorphic functions.   

\begin{prop}\label{schwinger} Let $\Upphi$ be a collection of fields on a pseudo-Hilbert space $\hil$ with inner product $(\cdot\,,\cdot)$. Suppose that the fields have a common dense domain $D$ such that $\f(f)D\subset D$, and they are covariant under a pseudo-unitary Poincar\'e representation $\un$ that satisfies the spectral condition introduced in Axiom \ref{spectrum}. The operators $\un(a,L)$ are defined on $D$, and $\un(a,L)D=D$. Let $\vac\in D$ be a Poincar\'e invariant vector. The vacuum expectation value
\[(\vac,\Fd{n_1}{\!i_1}\!(f_1)\dots\Fd{n_\ell}{\!i_\ell}\!(f_\ell)\vac),
\]
which is a separately continuous multilinear functional of $f_1,\dots,f_\ell\in\sch(\mink)$, can be uniquely extended to a tempered distribution $\mathcal{W}_{i_1\dots i_\ell}$ on $\mink^\ell$. There is a $W\in\sch(\mink^{\ell-1})$ such that 
\eq{\label{Weqn}
\mathcal{W}(x_1,\dots,x_\ell)=W(x_1-x_2,\dots,x_{\ell-1}-x_\ell),
}
where we suppressed the field indices $i_1,\dots,i_\ell$ for simplicity. There is a holomorphic function $\mathbf{W}$ on $\big\{\,(z_1,\dots,z_{\ell-1})\in\mathbb{C}^{\ell-1}\,|\,z_m=\xi_m+\im\eta_m,\,\xi_m\in\mink,\,\eta_m\in V^+,\,m=1,\dots,\ell-1\big\}$ such that
\[
W(\xi_1,\dots,\xi_{\ell-1})=\lim_{\eta_1,\dots,\eta_{\ell-1}\to 0} \mathbf{W}(\xi_1+\im\eta_1,\dots,\xi_{\ell-1}+\im\eta_{\ell-1}),
\] 
where the limit is meant in the sense of convergence in $\sch^\prime(\mink^{\ell-1})$. Define the Euclidean correlation functions $\mathbf{S}$ on the collection of real vectors $x_m^{}=(x^0_m,\bs{x}^{\phantom{0}}_m)\in\mathbb{R}^4$ with $x^0_m>0$ by   
\[
\mathbf{S}(x_1,\dots,x_{\ell-1})=\mathbf{W}(z_1,\dots,z_{\ell-1}),\;\;\;\;\;z_m^{}=(\im x_m^0,\bs{x}^{\phantom{0}}_m),\;\;\;m=1,\dots,\ell-1.
\]
If the field $\f^{(n_m)}$ transforms in the $\oplus_{s=1}^{\kappa_m}(\mathrm{j}_s^m,\mathrm{k}_s^m)$ representation of $S\!L(2,\mathbb{C})$, then  
\[
\Mx{Q}{1}{\!i_1j_1}\!(A,B)\dots\Mx{Q}{\ell-1}{\!\!\!\!\!\!\!i_{\ell-1}j_{\ell-1}}\!(A,B)\,\mathbf{S}_{j_1\dots j_{\ell-1}}(x_1,\dots,x_{\ell-1})=\mathbf{S}_{i_1\dots i_{\ell-1}}\big(R(A,B)x_1,\dots,R(A,B)x_{\ell-1}\big),
\]
where $R(A,B)$ is the four dimensional rotation that the covering map $S\!U(2)\times S\!U(2)\to S\!O(4)$ assigns to $(A,B)$, and $Q^{(m)}$ is the $\oplus_{s=1}^{\kappa_m}(\mathrm{j}_s^m,\mathrm{k}_s^m)$ representation of $S\!U(2)\times S\!U(2)$. This equality is valid if the arguments $x_m$ as well as $R(A,B)x_m$ are all in the domain of $\mathbf{S}$. In fact, $\mathbf{S}$ can be extended to a larger domain so that it remains covariant, but we will not need this extension. The function $\mathbf{S}$ uniquely determines the distribution $\mathcal{W}$. 
\end{prop}
For convenience it is customary to write an equality of distributions as if they were functions. We followed this practice in the above proposition, in which $\mathbf{W}$ and $\mathbf{S}$ are ordinary functions, but $\mathcal{W}$ and $W$ are distributions, so they should be smeared with test functions. The reader can easily supply the required smearing. For example, Eq.~\eqref{Weqn} actually means
\[
\mathcal{W}(f)=W(\hat{f}),\;\;\;\hat{f}(x_1,\dots,x_{\ell-1})=\smallint\d^4y\,f(x_1+\dots+x_{\ell-1}+y,x_2+\dots+x_{\ell-1}+y,\dots,x_{\ell-1}+y,y).
\]  

A scale invariant quantum field theory is defined as in Definition \ref{deflcqft}, but now $\un$ is a pseudo-unitary representation of dilations and Poincar\'e transformations such that
\[
\un(\lambda)\,\Fd{n}{i}(f)\un(\lambda)^{-1}=\lambda^{s_n-4}\Fd{n}{\!i}\,\big(\{\lambda\}f\big)
\]
on $D$, where $\big(\{\lambda\}f\big)(x)\coloneqq f\big(\lambda^{-1}x\big)$ and $\un(\lambda)$ is the operator representing a scale transformation by $\lambda$. The number $s_n$ is called the scaling dimension of $\Upphi^{(n)}$. Or in the unsmeared notation: 
\[
\un(a,L,\lambda)\,\Fd{n}{i}(x)\un(a,L,\lambda)^{-1}=\lambda^{s_n}\Mx{S}{n}{ij}(L^{-1})\,\Fd{n}{\!j}(\lambda\,Lx+a), 
\]
where $(a,L,\lambda)$ is a scale transformation by $\lambda$ followed by a Poincar\'e transformation $(a,L)$. 

Note that if $\Upphi^{(n)}$ transforms in an irreducible representation of the group of dilations and Lorentz transformations, then the transformation rule has the postulated form by Schur's lemma. Matrices that are not diagonal in their Jordan normal form generate indecomposable representations of dilations. There are no fields of such transformation properties in the models we will propose for a potential application of our results. So for our purposes, we can disregard the possibility of such transformation rules. Note that for the invariance of the vacuum state, it is enough if $\vac$ changes by a complex phase. Therefore dilations could be represented nontrivially on $\vac$. We could easily account for this possibility, but it will not be included, since the scale invariant theories we consider are expected to be conformal with a conformally invariant vacuum state, and if the conformal symmetry is already present in the indefinite metric formulation, simplicity of the symmetry group implies that dilations must act on $\vac$ trivially.   

When we analyze the implications of covariance under pseudo-unitary representations, the metric operator will be helpful. Its existence is established by the following result:\cite{Bogolubov}
\begin{prop}
If $\hil$ is a pseudo-Hilbert space, then there is an invertible self-adjoint bounded linear operator $\eta$, called the \emph{metric operator}, such that the inner product is $(\cdot\,,\cdot)=\langle\cdot\,,\eta\cdot\rangle$, where $\langle\cdot\,,\cdot\rangle$ is the Hilbert space scalar product of $\hil$. The inverse of $\eta$ is also bounded.
\end{prop}  

\section{REPRESENTATIONS OF SPACETIME SYMMETRIES}\label{Representation}

In this section we specify the pseudo-unitary representations of dilations and Poincar\'e representations used in our discussion. If the metric operator $\eta$ commutes with translations $\un(a)$, then $a\mapsto\un(a)$ is a unitary representation because
\[\langle\vf,\eta\psi\rangle=(\vf,\psi)=(\un(a)\vf,\un(a)\psi)=\langle\un(a)\vf,\eta\un(a)\psi\rangle=\langle\un(a)\vf,\un(a)\eta\psi\rangle\] 
for any $\vf\in D$ and $\psi\in D\cap(\eta^{-1}D)$, where $D$ is the common dense domain of $\un(a)$, so $\un(a)$ extends to a unitary operator.

Consider a continuous pseudo-unitary representation of the group of dilations and Poincar\'e transformations on a separable Hilbert space $\hil$ such that the representing operators are continuous, the spectral condition is satisfied, and the metric operator commutes with translations. Such a representation is unitary equivalent to a pseudo-unitary representation $\un$ of the following form. The operators $\un(\lambda,a,L)$ are continuous operators on $\hil$, which is the orthogonal direct sum of some or all of the Hilbert spaces $\hil_0$, $\hil_1$, and $\hil_\mathfrak{c}$. These subspaces are invariant under $\un$. The indices $1$ and $\mathfrak{c}$ indicate the components referred to as the one-particle component and the component of continuous mass spectrum. Translations are represented trivially on $\hil_0$. The other two components $\hil_1$ and $\hil_\mathfrak{c}$ are the spaces of square integrable functions $\varphi_1:\mathbb{R}^3\to H_1$ and $\varphi_\mathfrak{c}:\mathbb{R}^+\times\mathbb{R}^3\to H_\mathfrak{c}$ with norms 
\eq{\label{norms}
\|\varphi_1\|^2=\midint_{\mathbb{R}^3}\sfr{\d^3\bs{p}}{|\bs{p}|}\,\|\varphi_1(p)\|_1^2,\;\;\;\;\;\;\;\|\varphi_\mathfrak{c}\|^2=\midint_{0}^\infty\sfr{\d m}{m}\midint_{\mathbb{R}^3}\sfr{\d^3\bs{p}}{p^0}\,\|\varphi_\mathfrak{c}(p)\|_\mathfrak{c}^2,
}
where $p=(p^0,\bs{p})$ with $p^0=|\bs{p}|$ in the first equality and $p^0=\sqrt{\bs{p}^2+m^2}$ in the second equation. $H_1$ and $H_\mathfrak{c}$ are Hilbert spaces with norms $\|\cdot\|_1$ and $\|\cdot\|_\mathfrak{c}$, respectively.

The inner product on $\hil$ takes the form
\[(\varphi,\psi)=(\varphi_0,\psi_0)_0+\midint_{\mathbb{R}^3}\sfr{\d^3\bs{p}}{p^0}\,\langle\varphi_1(p),\eta_1(p)\psi_1(p)\rangle_1+\midint_{0}^\infty\sfr{\d m}{m}\midint_{\mathbb{R}^3}\sfr{\d^3\bs{p}}{p^0}\,\langle\varphi_\mathfrak{c}(p),\eta_\mathfrak{c}(p)\psi_\mathfrak{c}(p)\rangle_\mathfrak{c},\]
where $\langle\cdot\,,\cdot\rangle_i$ are the Hilbert space scalar products on $H_i$, and $\eta_i(p)$ are metric operators on $H_i$ such that $p\mapsto\eta_i(p)$ are weakly Borel, $i=1,\mathfrak{c}$. See Appendix for the definition of weakly Borel functions. Here $(\cdot\,,\cdot)_0$ is a nondegenerate inner product on $\hil_0$. The vectors $\varphi_i$, $i=0,1,\mathfrak{c}$, denote the components of $\varphi$ in the direct sum $\hil=\hil_0\oplus\hil_1\oplus\hil_\mathfrak{c}$. If the representation defined on $\hil_i$ does not occur in $\un$, then the corresponding term in the direct sum and the inner product is omitted. 

The action of $\un$ on $\hil_1$ and $\hil_\mathfrak{c}$ is given by
\eq{\label{MultRule}
\big(\un(a,g)\varphi_i\big)(p)=\lambda
 e^{\im p\cdot a}Q_i(p,g)\,\varphi_i(g^{-1}p),\;\;\;\;\;i=1,\mathfrak{c},
}
where $Q_i(p,g)$ are continuous operators on $H_i$, and $p\mapsto Q_i(p,g)$ are weakly Borel. Here we used a single letter $g$ for an element $(L,\lambda)$ of the group $G$ of dilations and Lorentz transformations. These notations indicate the generality of the result. Indeed, if $G$ is a locally compact Hausdorff group, $\un$ is a continuous pseudo-unitary representation of a semidirect product $\trans\rtimes G$ such that the operators $\un(a,g)$ are continuous, and the metric operator commutes with translations, then
\eq{\label{uni_general}
\big(\un(a,g)\vf\big)(p)=\sqrt{\sfr{\d\mu(g^{-1}p)}{\d\mu(p)}}\,e^{\im p\cdot a}Q(p,g)\,\vf(g^{-1}p)
}
on a direct integral $\hil=\smallint^\oplus_{\transd}\d\mu(p)H_p$, where $p\mapsto Q(p,g)$ is weakly Borel, and the function under the square root is the Radon-Nikodym derivative of $E\mapsto\mu(g^{-1}E)$ with respect to $\mu$. Direct integrals and related concepts are introduced in the Appendix, where we also sketch the argument that leads to the expression in Eq.~\eqref{uni_general}. The previously given form of representations of dilations and Poincar\'e transformations in terms of operators on the direct sum of at most three Hilbert spaces originates in the fact that the support of $\mu$ admits a partition into at most three sets in the direct integral. The representation property $\un(g)\un(h)=\un(gh)$ implies that for all $g,h\in G$,  
\eq{\label{MultRuleGen}
Q(p,g)\,Q(g^{-1}p,h)=Q(p,gh)}
is satisfied for almost all $p$. It follows from the pseudo-unitarity of $\un(g)$ that for all $g\in G$, 
\eq{\label{isom0}Q(gp,g)^*\eta(gp)Q(gp,g)=\eta(p)}
for almost all $p$.

The zero measure set of momenta $p$ for which Eq.~\eqref{MultRuleGen} fails to hold may vary with $g$ and $h$ so that the measure of their union is not zero. Therefore we cannot conclude automatically that for almost all $p$ the equality holds for all $g$ and $h$. But we have some freedom in the choice of $Q(p,g)$. For instance, changing it on a set of $\mu$-measure zero at a fixed $g$ does not affect the operator $\un(g)$. Even if this freedom is enough to achieve that Eq.~\eqref{MultRuleGen} holds for all $p$ and $g$, it is not obvious that we can simultaneously guarantee the joint measurability of the matrix elements of $Q(p,g)$ as the function of $p$ and $g$.

A similar problem arises when we try to establish an upper bound on $\|Q(p,g)\|_p$, where $\|\cdot\|_p$ is the operator norm induced by the norm in $H_p$. The norm in $\hil$ will be denoted by $\|\cdot\|$. Weak continuity of $\un$ implies that the matrix element $g\mapsto\langle\vf,U(g)\psi\rangle$ is a bounded function on any compact subset of $G$ for all $\vf,\psi\in\hil$. By the uniform boundedness principle, so is $g\mapsto\|\un(g)\|$. Since $\|\un(g)\|$ is the $\mu$-essential supremum of $\|Q(p,g)\|_p$, it is always possible to choose $Q(p,g)$ without changing $\un(g)$ so that $\|Q(p,g)\|_p\leq\|\un(g)\|$ for all $p$ and $g$. But it is not obvious that this choice does not interfere with the properties discussed in the former paragraph. We shall not investigate under what assumptions it is possible to choose the operators $Q(p,g)$ so that they have all the desirable measure theoretic properties. Instead, we assume that it is achievable for the representations encountered in field theories. The requirements are summarized in the following definition:
\begin{defin}\label{DefNonsing} Let $\trans\rtimes G$ be the semidirect product of the translation group $\trans$ and some other topological group $G$. If $\un$ is a weakly continuous representation of $\trans\rtimes G$ such that the operators $\un(a,g)$ are bounded and given by Eq.~\eqref{uni_general} on $\smallint^\oplus_{\transd}\d\mu(p)H_p$, then for all $g,h\in G$ equation \eqref{MultRuleGen} holds for almost all $p$, and $\mu\mathrm{-ess}\sup_{p\in\transd}\|Q(p,g)\|_p\leq\|\un(g)\|$. We say that $\un$ is \emph{nonsingular} if $Q$ can be chosen without affecting $\un$ so that Eq.~\eqref{MultRuleGen} and $Q(p,1)=\mathbbm{1}$ are satisfied everywhere, the map $(p,g)\mapsto Q(p,g)$ is weakly Borel, and $\|Q(p,g)\|_p\leq\|\un(g)\|$ for all $p$ and $g$. 
\end{defin}    
A measurable unitary representation of a locally compact Hausdorff group is continuous. See Appendix IV in [\citeonline{Wightman}] for a proof, which can be used to show that $h\mapsto Q(p,h)$ is a continuous representation of the stabilizer $G_p$ of $p$ if the operators $Q(p,h)$ satisfy the conditions in the above definition.      

\section{SOME GENERAL CONSEQUENCES OF COVARIANCE}
\label{General}
First we derive a simple expression for the states that the field operators create from the vacuum. The derivation uses only the representation of translations, which is unitary by the assumption that the metric operator commutes with it, so the indefinite inner product does not play any role.   
\begin{theor}\label{matrix_integral}
Let $\f$ be a field on $\smallint^\oplus_{\transd}\d\mu(p)H_p$, where $\mu$ is a Borel measure and $H_p$ are separable Hilbert spaces with scalar product $\langle\,,\rangle_p$. Suppose that $\f$ is covariant under the representation $(\un(a)\vf)(p)=e^{\im a\cdot p}\vf(p)$ of the translation group $\trans$. Let $\vac$ be in the common domain of the field operators $\f(f)$ and $\un(a)\vac=\vac$ for all $a\in\trans$. Then there are vectors $\Gamma(p)\in H_p$ such that $p\mapsto\Gamma(p)$ is weakly Borel and for any $f\in\sch(\mink)$,
\[
\big(\f(f)\vac\big)(p)=\hat{f}(p)\,\Gamma(p).
\]  
\end{theor}
\begin{proof}
Let $\sigma\subset\transd$ be any component of the partition of $\transd$ in terms of which the direct integral is defined. It is enough to prove that $\big(\mathcal{P}\f(f)\vac\big)(p)=\hat{f}(p)\,\Gamma(p)$ with some $\Gamma$, where $\mathcal{P}$ is the orthogonal projection to the subspace of vectors $\vf\in\smallint^\oplus_{\transd}\d\mu(p)H_p$ for which $\vf(p)=0$ if $p\notin\sigma$. This subspace is $L^2(\transd\!,H,\tilde{\mu})$ with scalar product $\langle\vf,\psi\rangle=\smallint_{\transd}\d\tilde{\mu}(p)\,\langle\vf(p),\psi(p)\rangle_H$, where $H$ is a separable Hilbert space with scalar product $\langle\cdot\,,\cdot\rangle_H$, and $\tilde{\mu}(E)\coloneqq\mu(E\cap\sigma)$. 

The map $(f,g)\mapsto\langle\f(\bar{f})\vac,\f(g)\vac\rangle$ is separately continuous on $\sch(\mink)\times\sch(\mink)$, so by the Schwartz kernel theorem, it uniquely extends to a continuous linear map $\mathcal{W}$ on $\sch(\mink\times\mink)$. Translational covariance of $\f$ and the invariance of $\vac$ imply that $\mathcal{W}(x,y)=\mathcal{W}(x-a,y-a)$. Therefore there is a distribution $W\in\sch(\mink)$ such that $\mathcal{W}(x,y)=W(x-y)$. Since $0\leq\|\f(f)\vac\|^2=W(\bar{f}*f_-)$, where $f_-(x)=f(-x)$, the distribution $W$ is positive definite. By the Bochner-Schwartz theorem, $W$ is the Fourier transform of a tempered Borel measure $\nu$ on $\minkd$, so $W(f)=\smallint_{\minkd}\d\nu(p)\,\hat{f}(p)$ and there is a number $r\geq0$ such that $\smallint_{\minkd}\d\nu(p)\,(1+|p|^2)^{-r}<\infty$.  

Using the Cauchy-Schwarz inequality, we get
\eq{\label{CS}|\langle\vf,\f(f)\vac\rangle|^2\leq\|\vf\|^2\,\|\f(f)\vac\|^2=\|\vf\|^2\,\smallint_{\minkd}\d\nu(p)\,|\hat{f}(-p)|^2}
for any $\vf\in L^2(\transd\!,H,\tilde{\mu})$ and $f\in\sch(\mink)$. Since $\nu$ is tempered, $\nu(B_k)<\infty$, where $B_k\subset\transd$ is the closed ball of radius $k\in\mathbb{N}$ centered at the origin. Define the following linear maps on the space of smooth functions supported within $B_k$:
\[
\psi_{\vf,k}:f\mapsto\langle\vf,\f(\check{f})\vac\rangle,
\]
where $\check{f}$ is the inverse Fourier transform of the extension of $f$ by zero to $\minkd$. Let $C(B_k)$ be the space of continuous functions on $B_k$ with the topology of uniform convergence. It follows from inequality \eqref{CS} that $\psi_{\vf,k}$ uniquely extends to a continuous map $\tilde{\psi}_{\vf,k}$ on $C(B_k)$. The Riesz-Markov representation theorem implies that there is a complex regular Borel measure $\rho_{\vf,k}$ on $B_k$ of finite total variation such that 
\[
\tilde{\psi}_{\vf,k}(f)=\smallint_{B_k}\d\rho_{\vf,k}(p)\,f(p).
\]

Let $\charf_E$ be the characteristic function of the Borel set $E\subset B_k$ and $\check{\charf}_E$ its inverse Fourier transform. Define $\d\alpha_{\vf,\psi}(p)\coloneqq\langle\vf(p),\psi(p)\rangle_p\d\mu(p)$. For any $\vf,\psi\in L^2(\transd\!,H,\tilde{\mu})$,
\[
\langle\vf,U(a)\psi\rangle=\smallint_{\transd}\d\alpha_{\vf,\psi}(p)\,e^{\im a\cdot p},\;\;\;\;\mbox{where}\;\;\alpha_{\vf,\psi}(F)=\langle\varphi,P(F)\psi\rangle.
\]
where $F\subset\transd$ is a Borel set and $P(F)$ is the corresponding spectral projection of the translation generator. We have 
\eq{\label{TransState}\begin{aligned}
\smallint_{\trans}\d a\,\check{\charf}_E(a)\langle\vf,\un(a)\f(f)\vac\rangle&=\smallint_{\trans}\d a\,\check{\charf}_E(a)\smallint_{\transd}\d\alpha_{\vf,\f(f)\vac}(p)\,e^{\im a\cdot p}=\smallint_{\transd}\d\alpha_{\vf,\f(f)\vac}(p)\smallint_{\trans}\d a\,\check{\charf}_E(a)e^{\im a\cdot p}\\&=\smallint_E\d\alpha_{\vf,\f(f)\vac}(p)=\langle P(E)\vf,\f(f)\vac\rangle,
\end{aligned}}
where Fubini's theorem allowed us to exchange the order of the integrations. The conditions for Fubini's theorem are satisfied since the Cauchy-Schwarz inequality implies that the total variation of $\alpha$ is a finite measure. On the other hand, using the covariance of $\f$ and the invariance of $\vac$, we can also write
\eq{\label{TransField}\begin{aligned}
\smallint_{\trans}\d a\,\check{\charf}_E(a)\,\langle\vf,\un(a)\f(f)\vac\rangle&=\smallint_{\trans}\d a\,\check{\charf}_E(a)\,\langle\vf,\f(\{a\}f)\vac\rangle=\smallint_{\trans}\d a\,\check{\charf}_E(a)\smallint_{B_k}\d\rho_{\vf,k}(p)\,e^{\im a\cdot p}\hat{f}(p)\\&=\smallint_{B_k}\d\rho_{\vf,k}(p)\,\hat{f}(p)\smallint_{\trans}\d a\,\check{\charf}_E(a)\,e^{\im a\cdot p}=\smallint_{E}\d\rho_{\vf,k}(p)\,\hat{f}(p),
\end{aligned}}
where Fubini's theorem was used to get the first expression in the second line. 

Let $(v_n)_{n\in\mathbb{N}}$ be an orthonormal topological basis in $H$ and $\xi\in L^2(\transd\!,\mathbb{C},\tilde{\mu})$ such that it vanishes almost everywhere on a Borel set $E\subset B_k$. Then $P(E)\,\xi v_n=0$ for any basis element $v_n$. For an arbitrary Borel set $F\subset B_k$, there is a sequence $f_m\in\sch(\mink)$ such that the supports of $\hat{f}_m$ is contained by $B_k$ and $\hat{f}_m$ converges to $\charf_E$ in $L^1(B_k,\mathbb{C},\rho_{\xi v_n,k})$. Taking the $m\to\infty$ limit of equations \eqref{TransState} and \eqref{TransField} with the substitution $\vf=\xi v_n$ and $f=f_m$, we get $\rho_{\xi v_n,k}(E\cap F)=0$. We conclude that $|\xi|\tilde{\mu}(E)=0$ implies that $|\rho_{\xi v_n,k}|(E)=0$, where $|\rho_{\xi v_n,k}|$ is the total variation of $\rho_{\xi v_n,k}$. The Radon-Nikodym theorem for complex measures tells us that there is a $|\xi|\tilde{\mu}$-integrable complex valued function $g_{\xi v_n,k}$ such that $\rho_{\xi v_n,k}=|\xi|\,g_{\xi v_n,k}\,\tilde{\mu}$, or $\rho_{\xi v_n,k}=\bar{\xi}\,G_{\xi,n,k}\,\tilde{\mu}$, where $G_{\xi,n,k}(p)=g_{\xi v_n,k}(p)\,\arg\xi(p)$ if $\xi(p)\neq 0$ and $G_{\xi,n,k}(p)=0$ otherwise.    

To show that $G_{\xi,n,k}$ can be changed on a set of $\tilde{\mu}$-measure zero if necessary so that it becomes independent of $\xi$, we choose an $f\in\sch(\mink)$ such that $\hat{f}|_{B_k}=1$. Then 
\[\xi\mapsto\langle\tilde{\xi}v_n,\f(f)\vac\rangle=\smallint_{B_k}\d\tilde{\mu}\,\bar{\xi}\,G_{\tilde{\xi},n,l}\]
is a continuous sesquilinear map on $L^2(B_k,\mathbb{C},\tilde{\mu})$, where $l>k$ and $\tilde{\xi}$ is the extension of $\xi$ to $\transd$ by zero. It follows from the Riesz representation theorem that there is a vector $\Gamma_{f,n,k}$ in $L^2(B_k,\mathbb{C},\tilde{\mu})$ such that
\[\langle\tilde{\xi}v_n,\f(f)\vac\rangle=\smallint_{B_k}\d\tilde{\mu}\,\bar{\xi}\,\Gamma_{f,n,k}\]
for all $\xi\in L^2(B_k,\mathbb{C},\tilde{\mu})$, so $G_{\tilde{\xi},n,l}=\Gamma_{f,n,k}$ almost everywhere on $B_k$. Thus $\Gamma_{f,n,k}$ can be chosen so that it is independent of $f$. Let $\Gamma_{n,k}$ be the collection of such choices. Note that $\Gamma_{n,l}|_{B_k}=\Gamma_{n,k}$ almost everywhere for any $l\geq k$.
For any $\xi\in L^2(\transd\!,\mathbb{C},\tilde{\mu})$, we have 
\[\langle\xi v_n,\f(f)\vac\rangle=\lim_{k\to\infty}\smallint_{B_k}\d\tilde{\mu}\,\hat{f}\,\bar{\xi}\,G_{\tilde{\xi}_k,n,k+1}=\lim_{k\to\infty}\smallint_{B_k}\d\tilde{\mu}\,\hat{f}\,\bar{\xi}\,\Gamma_n,\]
where $\xi_k=\xi|_{B_k}$, so $\tilde{\xi}_k=\xi\,\charf_{B_k}$, and $\Gamma_n\in L^2_\mathrm{loc}(\transd\!,\mathbb{C},\tilde{\mu})$ is defined by the relation $\Gamma_n|_{B_k}=\Gamma_{n,k}$.

From the above equation and
\[
\langle\xi v_n,\f(f)\vac\rangle=\smallint_{\transd}\d\tilde{\mu}(p)\,\bar{\xi}(p)\,\langle v_n,(\f(f)\vac)(p)\rangle_H,
\]
we get $\langle v_n,(\f(f)\vac)(p)\rangle_H=\hat{f}(p)\Gamma_n(p)$. Since $\big(P\f(f)\vac\big)(p)=\sum_n\langle v_n,(\f(f)\vac)(p)\rangle_H\,v_n$, the series of $\Gamma_n(p)\,v_n$ is strongly convergent in $H$ for almost all $p$. Define $\Gamma(p)$ as the limit of this series if it exists, that is, $\Gamma(p)=\sum_n\Gamma_n(p)\,v_n$, and for the sake of definiteness, let $\Gamma(p)=0$ at the elements $p$ of a zero $\tilde{\mu}$-measure subset of $\transd$ for which the series is divergent. More precisely, $\Gamma$ is the equivalence class of this function. 
\end{proof}
The purpose of the next lemma is to isolate the more technical parts of the argument that establishes our main result so that we can keep the proofs in the later sections of the paper as elementary as possible. Recall that we have not solved certain measure theoretic difficulties that arose in the study of pseudo-unitary representations. We just assume that such pathologies, if they are possible, are absent in physically relevant representations. Definition \ref{DefNonsing} encapsulates what we consider acceptable. In the discussion of covariant fields, we encounter complications of similar nature. Since we hardly have any insight into the structure of interacting fields, in this case it would be inappropriate to postulate that fields exhibiting such irregularities are not interesting from a physical point of view. That is why we decided to pay full attention to the measure theoretic details in this part of the analysis. Since it does not take much more elaboration to get the result for groups other than dilations and Poincar\'e transformations, the statement is formulated and proved for a class of groups which is more general than what we need. 

The lemma asserts that a certain relation holds for almost all elements of an orbit of a group $G$ in $\transd$ and for almost all elements of some subset of $G$. To make sense of these statements, we need to specify the measures with respect to which they are meant. The measure $\mu$ on $\transd$ will be quasi-invariant, so the measure on the orbits is simply $\mu$. As for $G$, it is natural to consider a Haar measure, the existence of which is guaranteed by the topological conditions on $G$. Then there exists a Haar measure on the stabilizer subgroups $G_p$ of the elements $p$ of $\transd$ as well, which naturally gives rise to a measure on the right cosets $gG_p$. Indeed, in terms of a left invariant Haar measure $\nu$ on $G_p$, the measure $\nu_p(E)\coloneqq\nu(g^{-1}E)$ is well-defined for any Borel set $E\subset gG_p$ since $\nu(g^{-1}E)$ is independent of the choice of $g$ by the invariance of $\nu$.              
\begin{lem}\label{DQ}
Let $\f$ be a field on $\smallint^\oplus_{\transd}\d\mu(p)H_p$, where $\mu$ is a Borel measure and $H_p$ are separable Hilbert spaces with scalar product $\langle\cdot\,,\cdot\rangle_p$. Suppose that $\f$ is covariant under a nonsingular pseudo-unitary representation $\un$ of the semidirect product of the translation group $\trans$ and $G$, where $G$ is a locally compact second-countable group. The representation of $G$ that transforms the field components will be denoted by $S$. Let $\un$ be put into the form given by Eq.~\eqref{uni_general} with the operators $Q(p,g)$ satisfying Eq.~\eqref{MultRuleGen} and $Q(p,1)=\mathbbm{1}$ everywhere. Say that there is a translation invariant vector $\vac$ in the common domain of the field operators $\f(f)$. Let $\Gamma_i(p)$ be the vectors in terms of which $\big(\f_i(f)\vac)(p)=\hat{f}(p)\,\Gamma_i(p)$ by Theorem \ref{matrix_integral}. Then for all $f\in\sch(\mink)$ and almost all $p$, it is true for almost all $k$ which are in the same orbit as $p$ that the equality
\eq{\label{u_transform_lemma}
\sqrt{\sfr{\d\mu(gp)}{\d\mu(p)}}\,(\widehat{\{g\}f})(k)\,\Gamma_i(k)=\hat{f}(p)\smallsum_jS_{ij}(g)\,Q(k,g)\,\Gamma_j(p)
}
holds for almost all $g$ in the right coset $g'G_p$ of the little group $G_p$ of $p$, where $g'p=k$. Furthermore, the metric operators $\eta(p)$ can be chosen so that for all $g$ and $p$,
\eq{\label{Qisometric}Q(gp,g)^*\eta(gp)\,Q(gp,g)=\eta(p).}
 \end{lem}
\begin{proof}
The identity $\langle\vf,\eta\f(f)\vac\rangle=\langle\un(g)\vf,\eta\f(\{g\}f)\vac\rangle$ implies that for any $\vf\in\smallint^\oplus_{\transd}\d\mu(p)H_p$, $f\in\sch(\mink)$, and $g\in G$,
\[\begin{aligned}\midint_{\!\!\transd\;}\d\mu(p)\,\hat{f}(p)\,&\langle\vf(p),\eta(p)\Gamma_i(p)\rangle_p\\&=\midint_{\!\!\transd\;}\d\mu(p)\sqrt{\sfr{\d\mu(g^{-1}p)}{\d\mu(p)}}\,(\widehat{\{g\}f})(p)\smallsum_jS_{ij}(g^{-1})\,\langle Q(p,g)\vf(g^{-1}p),\eta(p)\Gamma_j(p)\rangle_p.\end{aligned}\]
Changing the integration variable from $p$ to $gp$, we conclude that for all $g\in G$, 
\[
\sqrt{\sfr{\d\mu(gp)}{\d\mu(p)}}\,(\widehat{\{g\}f})(gp)\smallsum_jS_{ij}(g^{-1})\,Q(gp,g)^*\,\eta(gp)\Gamma_j(gp)-\hat{f}(p)\,\eta(p)\Gamma_i(p)
\]
is zero for almost all $p$, or by using Eq.~\eqref{Qisometric}, which holds for almost all $p$,
\eq{\label{def_d}d_i(p,g)\coloneqq\sqrt{\sfr{\d\mu(gp)}{\d\mu(p)}}\,(\widehat{\{g\}f})(gp)\,\Gamma_i(gp)-\hat{f}(p)\smallsum_jS_{ij}(g)\,Q(k,g)\,\Gamma_j(p)
}
is zero for almost all $p$.

Now we have to resolve a measure theoretic technicality. In order to proceed, it would be helpful to be able to assert that for almost all $p$, the equality $d_i(p,g)=0$ holds for all $g$. Such a reversal of the universal quantifiers would be possible if failure of this equation was not allowed, not even on a zero measure set of $p$. But in this case it is not automatic, since the set of momenta for which the equation fails to hold may vary with $g$ in such a way that the union of these sets over the uncountably many group elements may not be of zero measure.

Let $\rho$ be a $\sigma$-finite Borel measure on $G$. Take an orbit $\sigma$ of $G$ in $\transd$. Define $d_n$ by
\eq{\label{def_dn}
d_n:(p,g)\mapsto\langle v_n,d(p,g)\rangle_p}
if $p\in\sigma$. The vectors $v_n$ form an orthonormal (topological) basis in $H_p$, which is the same Hilbert space for all $p\in\sigma$. Let $d_n(p,g)\coloneqq 0$ if $p\notin\sigma$. The function $d(p,g)$ in Eq.~\eqref{def_dn} is defined by Eq.~\eqref{def_d}, with the field index $i$ suppressed. Suppose that $d_n$ are Borel functions on $G\times\transd$, which we will show shortly. We already know that
\eq{\label{int_orig}
\smallint_G\d\rho(g)\,\smallint_{\transd}\d\mu(p)\,|d_n(p,g)|=0
}
since the integral with respect to $p$ vanishes. By Fubini's theorem, we also have
\eq{\label{int_rev}
\smallint_{\transd}\d\mu(p)\,\smallint_G\d\rho(g)\,|d_n(p,g)|=0,
}
which is possible only if for almost all $p$, the equality $d_n(p,g)=0$ holds for almost all $g$. The measure of countably many zero measure sets is zero, so we can replace $d_n(p,g)$ by the vector $d(p,g)$ itself in this statement. This is what we want to achieve. 

First we define a measure on $G$. Let $H$ be a closed subgroup of $G$. According to Lemma 1.1 in [\citeonline{Mackey3}], there is a Borel set $B$ in $G$ such that $B$ intersects each right $H$ coset in exactly one point. By a theorem due to Lusin and Suslin, if $X$ and $Y$ are Polish spaces and $f:X\to Y$ is continuous, then $f(E)$ is Borel for any Borel subset $E$ of $X$ such that the restriction of $f$ to $E$ is injective. Since $\sigma$ is homeomorphic to the space of right cosets, it follows from this theorem that if $H$ is the little group of a fixed element $k\in\transd$, then the map $\sigma\to G$, $p\mapsto V_{p\leftarrow k}$ is Borel measurable, where $V_{p\leftarrow k}\in B$ such that $p=V_{p\leftarrow k}k$. There is one and only one such element for each $p\in\sigma$. The subgroup $H$ inherits the topological properties of $G$ sufficient for the existence of a Haar measure $\nu$. Let $\hat{\nu}$ be its extension to $G$, that is, $\hat{\nu}(E)\coloneqq\nu(E\cap H)$ for any Borel set $E\subset G$. We will also use the measure $\hat{\mu}(F)\coloneqq\mu(F\cap\sigma)$ on the Borel sets $F$ of $\transd$. Let $E$ be a Borel set of $G$ and define
\[
\rho(E)=\midint_{\!\!\transd}\d\hat{\mu}(p)\midint_{G}\d\hat{\nu}(g)\,\Uptheta_E(V_{p\leftarrow k}g),
\]
where $\Uptheta_E$ is the characteristic function of $E$. The $\hat{\nu}$-integral is well-defined, but we should also demonstrate that the result is a measurable function of $p$. It is enough to show that for each $n$ the map $(p,g)\mapsto\gamma_n(p,g)=\Uptheta_n(V_{p\leftarrow k}g)$ is $\hat{\mu}\times\hat{\nu}$-measurable, where $\Uptheta_n$ is a sequence of continuous functions that converges to $\Uptheta_E$ in $L^1(G,\hat{\nu})$ as $n$ goes to infinity. There is always such a sequence since $G$ is metrizable and $\hat{\nu}$ is a Borel measure. We adapt a known argument which proves that a real valued function on $\mathbb{R}^2$ is Borel if it is measurable in its first variable and continuous in the second one. There is a metric that induces the topology of $G$. Since $G$ is separable, for any $\ell>0$ it can be covered by countably many balls of radius $1/\ell$. Let us label these balls by natural numbers. Define $\gamma_n^\ell(p,h)\coloneqq\gamma_n(p,h_c)$, where $h_c$ is the center of the ball that contains $h$ and has the smallest label. The pointwise limit of $\gamma_n^\ell$ is $\gamma_n$ as $\ell\to\infty$. Since each $\gamma_n^\ell$ is $\hat{\mu}\times\hat{\nu}$-measurable, so is $\gamma_n$. By Beppo-Levi's theorem, $\rho$ is $\sigma$-additive. Therefore $\rho$ is a Borel measure on $G$. We also want $\rho$ to be $\sigma$-finite. We do not have to check this property since $\hat{\mu}$ and $\hat{\nu}$ are $\sigma$-finite measures, so they can be replaced by equivalent finite measures, in which case $\rho$ is even finite.

The Borel measurability of $d_n$ defined by Eq.~\eqref{def_dn} can be demonstrated as follows. Since $(p,g)\mapsto Q(p,g)$ is assumed to be weakly Borel and $\Gamma_i$ is also weakly Borel, we only have to check the measurability of the Radon-Nikodym derivative as the function of $(p,g)$. Of course it is measurable as the function of $p$ for each $g$, but we need joint measurability in order to apply Fubini's theorem, which was essential in the derivation of Eq.~\eqref{int_rev}. There is a quasi-invariant measure on any orbit such that its Radon-Nikodym derivative has this property.\cite{Mackey3} Since all quasi-invariant measures on an orbit are equivalent to each other, this is true for all of them because any of them is the product of any other and a measurable function. (The Radon-Nikodym derivative can be freely chosen on any zero measure subset of $\transd$ for all $g\in G$, so its joint measurability means that it becomes a Borel function on $\transd\times G$ by appropriate choices.)

Since $\un$ is pseudo-unitary with respect to $\eta$, for all $g\in G$ equation \eqref{Qisometric} is satisfied for almost all $p$. We would like to conclude that for almost all $p$, this equality holds for a large enough subset of $G$. We follow the same procedure as in the first part of the proof. This time we put the Borel function $d_n(p,g)=\langle v_n,[Q(gp,g)^*\eta(gp)Q(gp,g)-\eta(p)]v_n\rangle_H$, where $v_n$ are the elements of an orthonormal (topological) basis in $H$. We know that \eqref{int_orig} holds, from which \eqref{int_rev} follows by Fubini's theorem, so for almost all $p$ equation \eqref{Qisometric} holds for almost all $g$. This implies that for almost all $p$, it is true that for almost all $k\in\transd$, equation \eqref{Qisometric} holds for almost all $h$ in the stabilizer $G_p$ if $g$ is replaced by $gh$, where $g$ is a fixed element such that $gp=k$ and Eq.~\eqref{Qisometric} holds. (The measure on $G_p$ is a left Haar measure.) Using the multiplication rule $Q(ghp,gh)=Q(gp,g)\,Q(p,h)$, we get for almost all $p$ that $Q(p,h)^*\eta(p)\,Q(p,h)=\eta(p)$ for almost all $h\in G_p$, and so for all $h\in G_p$ since $h\mapsto Q(p,h)$ is weakly continuous on $G_p$ and a set whose complement is of zero Haar measure is dense. Equation \eqref{Qisometric} can be used to define all $\eta(k)$ in terms of $\eta(p)$. The relation $Q(p,h)^*\eta(p)\,Q(p,h)=\eta(p)$ guarantees that $\eta(k)$ is well-defined because it does not depend on the choice of $g$ that brings $p$ into $k$.
\end{proof}

\section{FIELDS SATISFYING THE FREE FIELD EQUATION}\label{Reeh-schlieder}
Let $\f$ be field in a covariant quantum field theory whose commutator is a $c$-number:
\[
[\f(x),\f(y)]=\im\Updelta(x-y)\mathbbm{1}.
\]
The field $\f^+$ is defined by the relation $\f^+(f)\coloneqq\f(\bar{f})^+$, where the superscript $+$ on the right hand side indicates the adjoint (involution) with respect to the possibly indefinite inner product $(\cdot\,,\cdot)$. Assume that translations commute with the metric operator. The Green's functions of $\f$, which are the vacuum expectation values of its polynomials, can be expressed in terms of $\Updelta$ and a constant $c$ given by $(\vac,\f(f)\vac)=c\smallint\!f$, where $\vac$ is the vacuum. To see this, take a monotonically increasing sequence of smooth functions $\theta_\ell:\minkd\to\mathbb{R}$ which converge pointwise to $\theta$, where $\theta(p)=0$ if $p^0\leq 0$ and $\theta(p)=1$ otherwise. Define $\f_\ell(f)=\check{\f}(\hat{f}\theta_\ell)$. (The Fourier transform of a field is defined through its matrix elements, which are distributions.) Theorem \ref{matrix_integral} implies that $\f_\ell(f)\vac\to\f(f)\vac-c\smallint\!f\,\vac$ strongly as $\ell\to\infty$. Note that $\f_{\ell-}(f)\coloneqq\f^+{}_\ell(\bar{f})^+=\check{\f}(\hat{f}\theta_{\ell-})$, where $\theta_{\ell-}(p)=\theta_\ell(-p)$, so the commutators of $\f_{\ell-}$ with $\f$ are determined by $\Updelta$. Furthermore, $\f_{\ell-}(f)\vac=0$ by the spectrum condition. The desired expression is obtained by iterating the following manipulations:
\[\begin{aligned}
(\vac,\f(f_1)\f(f_2)\dots\f(f_n)\vac)&=(\f^+(\bar{f}_1)\vac,\f(f_2)\dots\f(f_n)\vac)\\&=\lim_{\ell\to\infty}(\f^+{}_\ell(\bar{f}_1)\vac,\f(f_2)\dots\f(f_n)\vac)+c\smallint\!f\,(\vac,\f(f_2)\dots\f(f_n)\vac)\\&=\lim_{\ell\to\infty}(\vac,\f_{\ell-}(f_1)\f(f_2)\dots\f(f_n)\vac)+c\smallint\!f\,(\vac,\f(f_2)\dots\f(f_n)\vac)\\&=\lim_{\ell\to\infty}\big\{\,(\vac,[\f_{\ell-}(f_1),\f(f_2)]\dots\f(f_n)\vac)+\ldots\\&\phantom{=\lim_{\ell\to\infty}\big\{\,}+(\vac,\f(f_2)\dots[\f_{\ell-}(f_1),\f(f_n)]\vac)\,\big\}\\&\hspace{14.63em}+c\smallint\!f\,(\vac,\f(f_2)\dots\f(f_n)\vac).
\end{aligned}\]
Wightman's reconstruction theorem\cite{Streater} tells us that a covariant field theory in a positive definite inner product is determined up to an isomorphism by the Green's functions of the fields. In particular, if all the commutators are $c$-numbers for a subset of fields closed under the involution, then these fields and the representation of the spacetime symmetry are determined up to an isomorphism on the subspace generated by vectors that the polynomials of the fields create from the vacuum. There is a reconstruction theorem applicable to the case of an indefinite metric.\cite{Yngvason} Therefore it is useful to know under what conditions the commutators of fields are $c$-numbers.

For the derivation of such conditions, the following theorem will be helpful. It is a consequence of the positivity of the energy and the covariance of the fields, which imply that the Green's functions are boundary values of holomorphic functions.\cite{Streater} This remains true in an indefinite metric.\cite{Kugo}
\begin{theor*}{\rm(Reeh-Schlieder)} In a covariant quantum field theory in a possibly indefinite inner product, we have 
\[
\overline{\mathcal{P}(O)\vac}=\hil
\]
for any nonempty open subset $O$ of spacetime, where $\vac$ is the vacuum and $\hil$ is the Hilbert space on which the field operators are defined. 
\end{theor*}
Sometimes one can prove that a local field satisfies a differential equation when it acts on the vacuum. One of the most important consequences of the Reeh-Schlieder theorem is that in this case the field satisfies the differential equation as an operator equation. More generally, we have the following corollary of the Reeh-Schlieder theorem, whose proof is available in many sources,\cite{Streater,Kugo} but since it is very short, it will also be included. 
\begin{prop}\label{reehcor} The causal complement $O'$ of any open subset $O$ of spacetime is defined as the interior of the set of points which are spacelike separated from every element of $O$. In any local covariant quantum field theory with a possibly indefinite inner product, if $O'$ is not empty and $A\in\mathcal{P}(O)$, then $A\vac=0$ implies $A=0$. In particular, if $O$ is bounded, then the only element of $\mathcal{P}(O)$ that annihilates the vacuum is the zero operator.   
\end{prop}
\begin{proof} Take a vector $\vf$ in the common domain $D$ of the elements of $\mathcal{P}$, and let $P'$ be any element of $\mathcal{P}(O')$. Then for any $A$ in $\mathcal{P}(O)$ satisfying $A\vac=0$, we have 
\[
   (P'\vac,A^+\vf) = (AP'\vac,\vf) = (P'A\vac,\vf) = 0,
\]
where the commutativity of the elements $\mathcal{P}(O)$ and $\mathcal{P}(O')$ was used. By the Reeh-Schlieder theorem, $\mathcal{P}(O')\vac$ generates the entire Hilbert space $\hil$, so $(P'\vac,A^+\vf)=0$ for all $P'$ in $\mathcal{P}(O')$ implies $A^+\vf=0$. Then $(A\psi,\vf)=(\psi,A^+\vf)=0$ for any $\psi$ and $\vf$ in $D$. Since $D$ is dense in $\hil$, $A\psi=0$.
\end{proof}

Later we will encounter fields in a local covariant quantum field theory that satisfy the wave equation on the vacuum: $\f(\Box f)\vac=0$. Then the proposition implies that $\Box\f=0$. It is sometimes not appreciated in the literature that this does not automatically mean that the field is free. By definition, a free field is also required to have the same commutator as a Fock free field. In particular, the commutator is a $c$-number. There are theorems that establish this property for nonzero mass,\cite{Jost,SchroerDiplom} in which case the wave equation is replaced by the Klein-Gordon equation, or for zero mass in a positive definite inner product.\cite{Pohlmeyer} Here we quote a more general result, which applies to massless fields in an indefinite metric:
\begin{theor*}{\rm(Greenberg-Robinson)} Let $\f$ be a field of a local covariant quantum field theory in a possibly indefinite metric. Assume that translations commute with the metric operator. If the support of $\check{\f}$ excludes a neighborhood of a spacelike point in momentum space, then the commutator of $\f$ is a $c$-number.
\end{theor*}
The proof is in [\citeonline{Antonio,Robinson,GreenbergSuppP}]. The applicability of the argument to the indefinite case was noted in [\citeonline{Nakanishi}].

\section{SCALE INVARIANCE AND FIELD EQUATION}
\label{Final}
Using the analytic properties of the two-point function, first we give a very short proof for an already known theorem:\cite{WeinbergFree} 
\begin{theor}\label{weinberg_theorem} Let $\f$ be a field in a covariant quantum field theory which is also scale invariant. The inner product may be indefinite. Assume that $\f$ transforms in the Lorentz representation $(j,0)$ or $(0,j)$ and its scaling dimension is $j+1$. Let $\vac$ be the vacuum. Then the two-point function $\langle\vac,\f^+(x)\f(y)\vac\rangle$ is unique up to a multiplicative constant and it satisfies the wave equation. 
\end{theor}
\begin{proof}
The two-point function $\langle\vac,\f^+(x)\f(y)\vac\rangle$ transforms in the $(j,j)$ representation of the Lorentz group. By Proposition \ref{schwinger}, the corresponding Euclidean correlation function $\mathbf{S}$ transforms in the $(j,j)$ representation of $\mathrm{Spin}(4)=S\!U(2)\times S\!U(2)$, and it has scaling dimension $2j+2$. The direct product of dilations and $\mathrm{Spin}(4)$ acts transitively on the domain of $\mathbf{S}$, so $\mathbf{S}$ is determined by its value $\mathbf{S}(x)$ at one point. The covariance of $\mathbf{S}$ implies that $\mathbf{S}(x)$ is invariant under the stabilizer $\mathrm{Spin}_x$ of $x$, which is isomorphic to $S\!U(2)$. The restriction of a $(j,k)$ representation of $\mathrm{Spin}(4)$ to it is the sum of the representations from spin $|j-k|$ to $j+k$. Since the trivial representation occurs only once in the restriction of $(j,j)$ to $\mathrm{Spin}_x$, there is only one $\mathbf{S}$ up to a constant invertible matrix acting on the components of $\mathbf{S}(x)$. 

Define the $n$-index tensor field $\mathbf{W}(z)$ on $\big\{\,z\in\mathbb{C}\,|\,z=x+\im\eta,\,x\in\mink,\,\eta\in V^+\big\}$ as follows. Take $\partial_{\mu_1}\dots\partial_{\mu_n}1/z^2$, where $n=2j$ and $z^2=z\cdot z$ is the Lorentz square of $z$. The differentiations are meant with respect to $x$. Adding terms with appropriate coefficients in which one or more pairs of the indices are contracted with the Minkowski metric tensor $g$, such as $g_{\mu_1\mu_2}\partial_\nu\partial^\nu\partial_{\mu_2}\dots\partial_{\mu_n}1/z^2$, one obtains a completely symmetric tensor field $\mathbf{W}$ which is traceless in any pair of the indices. Since $\mathbf{W}$ is analytic and $\mathbf{S}(x)=\mathbf{W}(\im x^0,\bs{x})$ has the Euclidean and scale transformation properties specified in the previous paragraph, the components of $W(x-y)=\lim_{\eta\to 0}\mathbf{W}(x-y+\im\eta)$ are mapped into those of the two-point function $\langle\vac,\f^+{}_{\!\!\!i\,\,}(x)\f_{\!j}(y)\vac\rangle$ by a constant matrix $K_{ij,\mu_1\dots\mu_n}$ at all $x$ and $y$. The limit that defines $W$ is meant in the sense of convergence in $\sch^\prime(\mink)$. Integrating by parts in $\smallint_\mink\d x\,\mathbf{W}(x+\im\eta)\Box f(x)$, where $f\in\sch(\mink)$ is an arbitrary test function, we get $\smallint_\mink\d x\,\mathbf{W}(x+\im\eta)\Box f(x)=0$ because $\Box 1/z^2=0$ on the domain of $\mathbf{W}$ and $\mathbf{W}$ is the linear combination of derivatives of $1/z^2$ with constant coefficients, so $\Box \mathbf{W}=0$. By taking the limit $\eta\to0$, we obtain $W(\Box f)=0$, that is, $\Box W=0$.      
\end{proof}
In the second part of the proof, we calculated the two-point function in order to demonstrate that it satisfies the wave equation. If the inner product is positive definite, this step could be replaced by a reference to Weinberg's construction\cite{WeinbergMassless} of free fields, which shows that there are covariant fields transforming in any $(j_1,j_2)$ Lorentz representation, which are also covariant under dilations with scaling dimension $j_1+j_2+1$ and satisfy the wave equation. By the positive definiteness of the metric and Proposition \ref{reehcor}, we get $\Box\f=0$ if $\f$ is local.

If the metric is indefinite, then the field does not necessarily satisfy the wave equation even if the two-point function does. In this case, we analyze directly the properties of the fields that have nonvanishing matrix elements between the vacuum and one-particle states.  
\begin{theor}\label{main_theorem}
Let $\f$ be a field in a covariant quantum field theory which is also scale invariant, with a nonsingular representation of dilations and Poincar\'e transformations. Assume that the possibly indefinite metric operator $\eta$ commutes with translations. The orthogonal projections to the one-particle subspace and the subspace of continuous mass spectrum are denoted by $\mathcal{P}_1$ and $\mathcal{P}_\mathfrak{c}$, respectively. Let $\f$ transform in the $(j_1,j_2)$ Lorentz representation, and let $s$ be its scaling dimension. Then 
\begin{enumerate}[i.]
\item if $\;\;\langle\f(f)\vac,\eta\,\mathcal{P}_1\f(f)\vac\rangle\neq0\;\;$  for some $f\in\sch(\mink)$, then $s\leq j_1+j_2+1$,
\item if $\;\;\mathcal{P}_\mathfrak{c}\f(f)\vac\neq0\;\;$ for some $f\in\sch(\mink)$, then $s>j_1+j_2+1$.
\end{enumerate}
In other words, if $\f$ interpolates between the vacuum and (i) physical one-particle states (one-particle states of ``nonzero norm''), then $s\leq j_1+j_2+1$, (ii) states of continuous mass spectrum, then $s>j_1+j_2+1$.
\end{theor}
\begin{proof}
Let $G$ be the group of dilations and Lorentz transformations. If $\un$ is a continuous pseudo-unitary representation of $G$ by continuous operators, then a unitary operator brings it into the form given by Eq.~\eqref{MultRule}. We assume that it has already been achieved by a suitable choice of the operators $Q(p,g)$ that they satisfy the regularity conditions by which nonsingular representations are defined. In particular, for all $g$ and $h$, equation \eqref{MultRuleGen} holds for all $p$ and not only for almost all $p$, which is what follows automatically from the representation property $\un(g)\un(h)=\un(gh)$. The decomposition of $\un$ is given in Section \ref{Representation}. The one-particle orbit is $\sigma_1=\partial V^+\setminus\{\,0\,\}$, and the orbit of continuous mass spectrum is $\sigma_\mathfrak{c}=V^+$. By Axiom \ref{vacuum}, the subspace furnishing a trivial representation of translations is spanned by the vacuum $\vac$. Theorem \ref{matrix_integral} guarantees that $\f(f)\vac$ can be written in the form $\big(\f_i(f)\vac\big)(p)=\hat{f}(p)\Gamma_i(p)$.

Let $p$ and $k$ be two momenta, both in either $\sigma_1$ or $\sigma_\mathfrak{c}$, such that Eq.~\eqref{u_transform_lemma} in Lemma \ref{DQ} is satisfied for almost all $h\in G_p$ if $g$ is replaced by $gh$, where $g$ is any group element for which $gp=k$. Let $h$ be a Lorentz transformation $\mathfrak{h}$ followed by a rescaling by $\lambda_h\in\mathbb{R}^+$. Then 
\[\sfr{\d\mu(hp)}{\d\mu(p)}=\lambda_h^{-2},\;\;\;\;\;\;\;(\widehat{\{h\}f})(hp)=\lambda_h^4\hat{f}(p),\]
where $\mu$ is either $\mu_1$ or $\mu_\mathfrak{c}$, which are the measures on $\sigma_1$ and $\sigma_\mathfrak{c}$, respectively, with respect to which we integrate in Eq.~\eqref{norms}. The representation $S$ in which the field transforms is expressed in terms of the Lorentz representation $D$ and the scaling dimension $s$ by $S(h)=\lambda_h^{4-s}D(\mathfrak{h})$. Then for any two $h_1,h_2\in G_p$ for which Eq.~\eqref{u_transform_lemma} is satisfied with $g$ replaced by $gh_1$ or $gh_2$, we have
\[
\lambda_{h_1}^{1-s}\smallsum_jD_{ij}(\mathfrak{g}\mathfrak{h}_1)\,Q(k,gh_1)\,\Gamma_j(p)=\lambda_{h_2}^{1-s}\smallsum_jD_{ij}(\mathfrak{g}\mathfrak{h}_2)\,Q(k,gh_2)\,\Gamma_j(p),
\]
where $\mathfrak{g}$ is the Lorentz transformation in $g$. By writing $h=h_2^{-1}h_1^{}$, using the equality $\lambda_h=\lambda_{h_1}/\lambda_{h_2}$ and the multiplication rules for $D$ and $Q$ as well as the validity of $Q(p,1)=\mathbbm{1}$ for all $p$, we conclude that for almost all $p$ in $\sigma_1$ and $\sigma_\mathfrak{c}$,  
\eq{\label{u_transform_final}
\lambda_h^{1-s}\smallsum_jD_{ij}(\mathfrak{h})\,Q(p,h)\,\Gamma_j(p)=\Gamma_i(p)
}
for almost all $h\in G_p$ with respect to a Haar measure, and therefore for all $h\in G_p$ because $h\mapsto Q(p,h)$ is a weakly continuous function on $G_p$ and the set whose complement is of zero Haar measure is dense. 

(i) By assumption,
\[\langle\f_i(f)\vac,\eta\,\mathcal{P}_1\f_j(f)\vac\rangle=\smallint_{\sigma_1}\d\mu_1(p)\,|\hat{f}(p)|^2\,\langle\Gamma_i(p),\eta(p)\Gamma_j(p)\rangle_1\] 
is not zero for some $f$, where $\langle\cdot\,,\cdot\rangle_1$ is the scalar product of the Hilbert space in which a state takes its value on $\sigma_1$. Therefore the matrix $\langle\Gamma_i(p),\eta(p)\Gamma_j(p)\rangle_1$ cannot be zero for almost all $p$. So there are momenta $p$ for which this matrix is not zero and, simultaneously, Eq.~\eqref{u_transform_final} holds for all $h$ in $G_p$, so
\[
\lambda_h^{2-2s}\smallsum_{lk}\bar{D}_{il}(\mathfrak{h})\,D_{jk}(\mathfrak{h})\,\langle \Gamma_l(p),\eta(p)\Gamma_k(p)\rangle_1=\langle\Gamma_i(p),\eta(p)\Gamma_j(p)\rangle_1.
\]
Let $\mathfrak{b}(\hat{\bs{n}},\theta)$ be a boost of rapidity $\theta$ in the direction given by the spatial unit vector $\hat{\bs{n}}$. If $s(\lambda)$ is the scale transformation $x\mapsto\lambda x$ of vectors $x\in\trans$, then $s(\lambda)\,\mathfrak{b}(\hat{\bs{p}},\ln\lambda)\in G_p$ for any $p\in\sigma_1$, where $\hat{\bs{p}}=\bs{p}/|\bs{p}|$ and $p=(|\bs{p}|,\bs{p})$. The matrix $\langle\Gamma_i(p),\eta(p)\Gamma_j(p)\rangle_1$ is an eigenvector of the linear operator $\bar{D}\big(\mathfrak{b}(\hat{\bs{p}},\ln\lambda)\big)\otimes D\big(\mathfrak{b}(\hat{\bs{p}},\ln\lambda)\big)$ with eigenvalue $\lambda^{2s-2}$. Note that the $\big(\frac{1}{2},0\big)$ representation of $\mathfrak{b}(\hat{\bs{n}},\ln\lambda)$ can be brought into the matrix
\[
\left(\begin{matrix}\lambda^{\frac{1}{2}}&0\\0&\lambda^{-\frac{1}{2}}\end{matrix}\right)
\] 
by a similarity transformation. The $(j_1,0)$ representation is the $2j_1$-fold symmetric product of $\big(\frac{1}{2},0\big)$, the complex conjugate of $(j_2,0)$ is $(0,j_2)$, and $(j_1,j_2)=(j_1,0)\otimes(0,j_2)$, so we conclude that the largest negative exponent with which $\lambda$ occurs in the matrix of $\bar{D}\big(\mathfrak{b}(\hat{\bs{p}},\ln\lambda)\big)\otimes D\big(\mathfrak{b}(\hat{\bs{p}},\ln\lambda)\big)$ is $-2(j_1+j_2)$, so $s\leq j_1+j_2+1$.         

(ii) We shall assume that \eqref{Qisometric} is satisfied for all $g$ and $p$. Lemma \ref{DQ} guarantees that this can always be achieved by a suitable choice of $\eta(p)$. It follows from Lemma \ref{DQ} that for almost all $k\in\sigma_\mathfrak{c}$, it is true for almost all $p\in\sigma_\mathfrak{c}$ that there is a rotation $\mathfrak{q}_p\in G_k$ such that 
\eq{\label{kappa_expression}
\Gamma_i(p)=\left(\sfr{\,m{}_{\phantom{0}}\!}{\,m_0\!}\right)^{s-1}\smallsum_jD_{ij}\big(\mathfrak{r}_p\mathfrak{b}_p\mathfrak{q}_p\big)\,Q\big(p,V_{p\leftarrow k}\mathfrak{q}_p\big)\,\Gamma_j(k),
}
where $m_0=\sqrt{-k\cdot k}$, $m=\sqrt{-p\cdot p}$. Choose a basis in $\transd$ in which $k=(m_0,\bs{0})$. The rotations $\mathfrak{r}_p\in G_k$ are some fixed elements such that $\mathfrak{r}_p(0,\hat{\bs{n}})=(0,\hat{\bs{p}})$ if $\bs{p}\neq 0$ and $\mathfrak{r}_p=1$ otherwise. Here $\hat{\bs{n}}$ is some spatial unit vector. In the previously introduced notations $\mathfrak{b}$ and $s$ for boosts and scale transformations, respectively,  
\[
\mathfrak{b}_p\coloneqq\mathfrak{b}\left(\hat{\bs{n}},\tanh^{-1}\!\sfr{\,|\bs{p}|{}^{\phantom{0}}\!}{\,p^0\!}\right),\;\;\;V_{p\leftarrow k}\coloneqq s\left(\sfr{\,m_0\!}{\,m{}_{\phantom{0}}\!}\right)\mathfrak{r}_p\,\mathfrak{b}_p.
\]
As indicated by the notation, $p=V_{p\leftarrow k}k$. The appearance of $\mathfrak{q}_p$ in \eqref{kappa_expression} is due to the measure theoretic technicality that arises in the proof of Lemma \ref{DQ}. Note that $\lambda_h=1$ for any $h\in G_k$. From now on, let $k$ be a momentum for which, in addition to relation \eqref{kappa_expression}, equation \eqref{u_transform_final} is also satisfied for all $h\in G_k$, with $p$ replaced by $k$. Almost every $k$ has this property, so such a choice is possible. Then $\mathfrak{q}_p$ can be eliminated by writing $D\big(\mathfrak{r}_p\mathfrak{b}_p\mathfrak{q}_p\big)=D\big(\mathfrak{r}_p\mathfrak{b}_p\big)\,D\big(\mathfrak{q}_p\big)$ and $Q\big(p,V_{p\leftarrow k}\mathfrak{q}_p\big)=Q\big(p,V_{p\leftarrow k}\big)\,Q\big(k,\mathfrak{q}_p\big)$, and using $\eqref{u_transform_final}$. 

A significant difference between the one-particle case and states of continuous mass spectrum is that the stabilizer $G_k$ of an element in $\sigma_\mathfrak{c}$ is $S\!U(2)$, which is compact. So the representation $q$ of $G_k$ defined by $q(h)\coloneqq Q(k,h)$ is a representation of a compact group, which is pseudo-unitary with respect to the inner product given by the metric operator $\eta(k)$. This implies that the orthogonal $\eta(k)$-invariant subspaces $H_{k\pm}$ on which $\eta(k)$ is strictly positive and negative definite are invariant under $q$. The proof is given in [\citeonline{Waw}], and it is based on the finiteness of the Haar measure $\nu$ of a compact group, which allows us to change the scalar product on $H_\mathfrak{c}$ to an equivalent one with respect to which $q$ is unitary: $\{u,v\}=\smallint_{G_k}\d\nu(h)\,\langle q(h)u,q(h)v\rangle_\mathfrak{c}$. Here $\langle\cdot\,,\cdot\rangle_\mathfrak{c}$ is the scalar product in the Hilbert space $H_\mathfrak{c}$ in which a state takes its value on $\sigma_\mathfrak{c}$. 

The subspace $\hil_\mathfrak{c}$ of continuous mass spectrum is the direct sum of two orthogonal $\un$-invariant subspaces on which $\eta$ is strictly positive or negative definite. To show this, consider the closed subspaces $\hil_\pm\subset\hil_\mathfrak{c}$ generated by vectors of the form
\[\vf_\pm(p)=Q(p,V_{p\leftarrow k})u_\pm(p),\]
where $u_\pm$ are elements of $L^2(\sigma_\mathfrak{c},H_{k\pm},\mu_\mathfrak{c})$ which vanish outside a compact set contained by $V^+$. Since $p\mapsto\mathfrak{b}_p$ is continuous, any compact set in $V^+$ is mapped into a subset of a compact set of $G$ by $V_{p\leftarrow k}$. So the norm of $Q(p,V_{p\leftarrow k})$ is a bounded function of $p$ on the support of $u_\pm$, which guarantees that above formula defines vectors in $\hil_\mathfrak{c}$. (It is assumed that the rotations $\mathfrak{r}_p$ are chosen so that $\varphi_\pm$ have the required measurability properties.) The subspaces $\hil_\pm$ are invariant under $\un$ because for any $\vf_\pm$ of the above form, we have 
\[
\big(\un(g)\varphi_\pm\big)(p)=Q(p,g)\,Q(g^{-1}p,V_{g^{-1}p\leftarrow k})\,u_\pm(g^{-1}p)=Q(p,V_{p\leftarrow k})\,Q(k,V_{p\leftarrow k}{}^{\hspace{-1em}-1}\;gV_{g^{-1}p\leftarrow k})\,u_\pm(g^{-1}p),
\]
and $Q(k,V_{p\leftarrow k}{}^{\hspace{-1em}-1}\;gV_{g^{-1}p\leftarrow k})$ is an operator representing an element of $G_k$, which maps $H_\pm$ onto itself. Let us check that they are also orthogonal. For an $h\in G_k$,
\[
Q(p,V_{p\leftarrow k})\,Q(k,h)\,u_\pm(p)=Q(p,V_{p\leftarrow k}hV_{p\leftarrow k}{}^{\hspace{-1em}-1}\;)\,Q(p,V_{p\leftarrow k})\,u_\pm(p).
\]     
Since $V_{p\leftarrow k}G_kV_{p\leftarrow k}{}^{\hspace{-1em}-1}\;=G_p$, the spaces $H_{p\pm}$ spanned by vectors $\vf_\pm(p)$ are invariant under the representation $Q(p,\cdot)$ of $G_p$ because $H_{k\pm}$ are invariant under $q$.  This representation is pseudo-unitary with respect to the inner product specified by $\eta(p)$, and by Eq.~\eqref{Qisometric} the metric operator $\eta(p)$ is strictly positive and negative definite on $H_{p+}$ and $H_{p-}$, respectively. The previously quoted theorem tells us that $H_{p\pm}$ are orthogonal. We conclude that $\hil_\mathfrak{c}$ is the direct sum of the orthogonal subspaces $\hil_+$ and $\hil_-$.     

By the hypothesis of the theorem, the orthogonal projection of $\f(f)\vac$ to one of the subspaces $\hil_\pm$ must be nonzero for some $f$. Say that it is $\hil_+$. The other case is completely analogous. The reason why it is important that $\hil_+$ and its orthogonal complement are both invariant under $\un$ is that the orthogonal projection $\mathcal{P}_+$ to $\hil_+$ commutes with $\un$, so $\mathcal{P}_+\f_i(f)\vac$ can be written as
\[
(\mathcal{P}_+\f_i(f)\vac)(p)=\hat{f}(p)\left(\sfr{\,m{}_{\phantom{0}}\!}{\,m_0\!}\right)^{s-1}\smallsum_jD_{ij}\big(\mathfrak{r}_p\mathfrak{b}_p\big)\,Q(p,V_{p\leftarrow k})\,P_+\Gamma_j(k),
\]
where $P_+$ is the orthogonal projection to $H_+$ in $H_\mathfrak{c}$. The definiteness of $\eta$ on $\hil_+$ implies that the scalar product
\[\begin{aligned}
\langle\f_i(f)\vac,\eta\,&\mathcal{P}_+\f_j(f)\vac\rangle\\&=\midint_{\sigma_\mathfrak{c}}\d\mu_\mathfrak{c}(p)\,\left(\sfr{\,m{}_{\phantom{0}}\!}{\,m_0\!}\right)^{2s-2}\,|\hat{f}(p)|^2\smallsum_{nl}\bar{D}_{in}\big(\mathfrak{r}_p\mathfrak{b}_p\big)\,D_{jl}\big(\mathfrak{r}_p\mathfrak{b}_p\big)\,\langle\Gamma_n(k),\eta(k)P_+\Gamma_l(k)\rangle_\mathfrak{c}
\end{aligned}\]
is not zero for some $f$. This integral can be nonzero for some $f\in\sch(\mink)$ only if the matrix $\langle\Gamma_n(k),\eta(k)P_+\Gamma_j(k)\rangle_\mathfrak{c}$ is not zero. But of course the integral should also be finite for any $f\in\sch(\mink)$. We can redefine the field by $\f_i(f)\to K_{ij}\f_j(f)$, where $K$ is a nondegenerate matrix $K$. Since $\mathfrak{r}_p$ are all in the stabilizer of a timelike vector $k$, which is isomorphic to $S\!U(2)$, by an appropriate $K$ the matrices $D\big(\mathfrak{r}_p\big)$ can always be brought into matrices which are unitary with respect to the standard scalar product on $\mathbb{C}^J$, where $J$ is the number of field components. Since $\eta(k)$ is strictly definite on $H_\pm$, the restriction of the $q(h)$ to $H_\pm$ gives a continuous representation of $S\!U(2)$, which is essentially unitary and consequently decomposes into the direct sum of irreducible representations, whose representation spaces are $\eta(k)$-orthogonal to each other. The Fourier transformation is a linear isomorphism from $\sch(\mink)$ onto $\sch(\minkd)$, so we obtain the condition that for any $\hat{f}\in\sch(\minkd)$,
\eq{\label{finiteness}
\midint_{0}^\infty\sfr{\d m}{m}\midint_{\mathbb{R}^3}\sfr{\d^3\bs{p}}{p^0}\,m^{2s-2}\,|\hat{f}(p)|^2\smallsum_{ikl}\bar{D}_{ik}\big(\mathfrak{b}_p\big)\,D_{il}\big(\mathfrak{b}_p\big)\,T_{kl}<\infty,\;\;\;\;\;\;\;T_{ij}=\smallsum_\alpha\overline{C}_{i\alpha}C_{j\alpha}, 
}
where $C_{i\alpha}$ are the components of the projection of $\Gamma_i(k)$ to a subspace $H_r$ of some spin-$r$ representation in a basis in which $\eta(k)|_{H_r}$ is the identity matrix. Here we used that $\langle\Gamma_i(k),\eta(k)P_+\Gamma_j(k)\rangle_\mathfrak{c}$ is the sum of terms similar to $T_{ij}$, all of which give positive contribution to the sum in the integrand of Eq.~\eqref{finiteness}, so they are separately finite and one of them is not zero. Equation \eqref{u_transform_final}, with the replacement of $p$ by $k$, tells us that $C_{i\alpha}$ are proportional to the Clebsch-Gordan coefficients of the trivial representation in the product of a spin-$r$ representation of $S\!U(2)$ and the representation furnished by the field components. Since each spin occurs once and only once between $|j_1-j_2|$ and $j_1+j_2$ in the restriction of $D$ to $S\!U(2)$, we have $|j_1-j_2|\leq r\leq j_1+j_2$, and the matrix $T_{ij}$ is determined up to a factor if the bases in $\mathbb{C}^J$ and in the subspace $H_r$ are fixed.

In order not to clutter the formulas with too many indices, we shall replace $(j_1,j_2)$ by $(k,l)$ in the following calculations. This will not cause any confusion since $k$ and $l$ no longer occur as summation indices. The vectors $v$ in $\mathbb{C}^J$ are written as tensors $v_{a^1\dots a^{2k},\ddot{b}^1\dots\ddot{b}^{2l}}$, where the indices $a^1,\dots,b^{2l}$ are spinor indices that can take two values, $0$ and $1$. The tensor $v$ is completely symmetric in both the $a$ and the $b$ indices. The dot on the $b$ indices is just a mnemonic for the $\big(0,\frac{1}{2}\big)$ transformation property of the tensor labeled by them, as opposed to the undotted indices, which refer to the components of $\big(\frac{1}{2},0\big)$ spinors. A similar labeling, a symmetric collection of indices $c^1,\dots,c^{2r}$, is introduced for vectors in the target space of states. The intertwiner $C_{i\alpha}$ is a tensor $C_{a^1\dots a^{2k},\ddot{b}^1\dots\ddot{b}^{2l},c^1\dots c^{2r}}$, completely symmetric in the $a$, $b$, and $c$ indices. (The $S\!U(2)$ representations are self-conjugate, so there is no need for a distinction between dotted and undotted indices. For the sake of definiteness, we chose undotted $c$ indices.) 

In the following formulas, $\delta_{a^1\ddot{b}^1}$ are the entries of the $2\times 2$ identity matrix, and $\ve$ is the antisymmetric $2\times 2$ matrix with $\ve_{01}=1$. $E$ is the same matrix as $\ve$. The reason why two different symbols are used for them will be clear shortly. Let the $\big(\frac{1}{2},0\big)$ Lorentz representation be given by matrices $S(L)$. The claim is that $C_{a^1\dots c^{2r}}$ is proportional to the following expression:
\eq{\label{Tproposed}
   \delta_{a^1\ddot{b}^1}\dots\delta_{a^{k+l-r}\ddot{b}^{k+l-r}}
   E_{a^{k+l-r+1}c^1}\dots E_{a^{2k}c^{k-l+r}} 
   \delta_{c^{k-l+r+1}\ddot{b}^{k+l-r+1}}\dots\delta_{c^{2r}\ddot{b}^{2l}},
}
where complete symmetrization is meant within the $a$, $b$, and $c$ indices (but there is no symmetrization involving say both $a$ and $b$ indices). To check the claim, first note that the invariance condition $\sum_{c\ddot{d}}S_{ac}(L)\bar{S}_{\ddot{b}\ddot{d}}(L)\delta_{c\ddot{d}} = \delta_{a\ddot{b}}$ means that $S(L)$ is unitary, which is true when $L$ is in the little group of $k$ by our choice of the basis in $\mathbb{C}^J$. The matrix $\ve$ is an intertwiner between the $S\!U(2)$ representation got by the restriction of $S$ and its complex conjugate. Suppose that the bases were chosen so that the elements $A$ of the little group of $k$ are represented by $S(A)=\exp[\im(\beta^1\sigma_1+\beta^2\sigma_2+\beta^3\sigma_3)]$, where $\beta^i$ are real numbers, and $\sigma_i$ are $\frac{1}{2}$ times the Pauli matrices. The complex conjugate of the matrix of $S(A)$ is the matrix of $\bar{S}(A)=\exp[\im(-\beta^1\sigma_1+\beta^2\sigma_2-\beta^3\sigma_3)]$. Since $\bar{S}(A)= \sigma_2 S(A) \sigma_2^{-1} = \ve S(A) \ve^{-1}$, as can be seen from the commutation relations of the Pauli matrices, $\sum_{cd}S_{ac}(A) S_{bd}(A) E_{cd} = E_{ab}$ for any $A$ in the little group, where $E_{ab}\coloneqq\sum_{\ddot{c}}\ve_{b\ddot{c}}\delta_{a\ddot{c}}$. (In fact this equation holds for any Lorentz transformation.) Now we can explain the somewhat pedantic distinction between $E$ and $\ve$. The former is an invariant vector in the representation space of $\big(\frac{1}{2},0\big)\otimes\big(\frac{1}{2},0\big)$, the latter is a map from the representation space of $\big(0,\frac{1}{2}\big)$ onto that of $\big(\frac{1}{2},0\big)$. Our proposal for $C_{i\alpha}$ is an $S\!U(2)$ invariant element of the representation space of the restriction of $(k,l)$ to $S\!U(2)$ times the spin $r$ representation, so the actual $C_{i\alpha}$ in the given bases is a constant multiple of this expression. The proposed formula in Eq.~\eqref{Tproposed} makes sense only if $0 \leq k+l-r \leq 2k$, and $k+l-r \leq 2l$, which are equivalent to the previously specified range $|l-k| \leq r \leq l+k$, from which the spin $r$ has to be selected. (If $k+l=r$, then the first delta factors are absent; if $l-k=r$, then there are no $E$ matrices; if $k-l=r$, then the last delta factors are missing.)

To get $T_{ij}$, we have to complex conjugate the expression for $C_{i\alpha}$ and contract the $c$-indices with those in the original formula. Under complex conjugation, undotted indices pick up a dot, while dotted indices lose it. In our basis it was more convenient to replace $i$ with a collection of $a$ and $b$ indices. Similarly, the index $j$ will be replaced by a collection of primed $a$ and $b$ indices. Using that $E$ is real, antisymmetric, and its square is proportional to the identity matrix, we find that the matrix of $T$ is proportional to
\[
   \delta_{a^1\ddot{b}^1}\dots\delta_{a^{k+l-r}\ddot{b}^{k+l-r}}
   \delta_{\ddot{a}'{}^1b'{}^1}\dots\delta_{\ddot{a}'{}^{k+l-r}b'{}^{k+l-r}}
   \delta_{a^{k+l-r+1}\ddot{a}'^{k+l-r+1}}\dots\delta_{a^{2k}\ddot{a}'{}^{2k}}
   \delta_{b'{}^{k+l-r+1}\ddot{b}{}^{k+l-r+1}}\dots\delta_{b'{}^{2l}\ddot{b}{}^{2l}},
\]
where complete symmetrization is meant in the $a$, $b$, $a'$, and $b'$ indices.   

Now we can calculate the effect of the boost $\mathfrak{b}_p$ on $T$. We have $S\big(\mathfrak{b}(\hat{\bs{n}},\beta)\big)=\exp[\im\beta\bs{n}\bs{\cdot}(-\im\bs{\sigma})]=\mathbbm{1}\cosh\frac{\beta}{2}+\hat{\bs{n}}\bs{\cdot}\bs{\sigma}\sinh\frac{\beta}{2}$, where $\bs{\cdot}$ denotes a sum over spatial indices. So  $S_{ac}(\mathfrak{b}_p)\,\bar{S}_{\ddot{b}\ddot{d}}(\mathfrak{b}_p)\,\delta_{c\ddot{d}}$ is the following matrix:
\[
\left(\begin{matrix}e^\beta&0\\0&e^{-\beta}\end{matrix}\right)=
\left(\begin{smallmatrix}\sfr{p^0+|\bs{p}|}{m}&0\\0&\sfr{m}{p^0+|\bs{p}|}\end{smallmatrix}\right).
\]
So $\mathfrak{b}_p$ transforms $T$ into a tensor whose component $T^0$ labeled by $a^1=\ldots=\dot{a}'{}^1=\ldots=\dot{b}^1=\ldots=b'{}^1=\ldots=0$ is a nonzero constant times $[(p^0+|\bs{p}|)/m]^{2(k+l)}$, regardless of the selected spin $r$. There are other components containing powers of $(p^0+|\bs{p}|)/m$ with lower exponents. 
 
We restore $j_1$ and $j_2$ in place of $k$ and $l$, and finish the discussion of the conditions for the finiteness of the expression in Eq.~\eqref{finiteness}. From the previous calculation we get 
\[\midint_{0}^\infty\sfr{\d m}{m}\midint_{\mathbb{R}^3}\sfr{\d^3\bs{p}}{p^0}(m^2)^{s-1-j_1-j_2}\big(p^0+|\bs{p}|\big)^{2(j_1+j_2)}\,|\hat{f}(p)|^2\, T^0 +\ldots<\infty,\]
where the ellipses indicates terms in which $\frac{1}{m}$ occurs with lower exponents than in the term containing $T^0$. The condition $s>j_1+j_2+1$ is necessary and sufficient for the existence of the integral for any test function $g$.
\end{proof}
Unlike in a positive definite inner product, in an indefinite metric the Lorentz transformation property does not fix the scaling dimension if the field creates physical one-particle states from the vacuum. Indeed, the gradient of a massless free scalar field has scaling dimension $s=2$, whereas the vector potential in the indefinite metric (Gupta-Bleuler) quantization\cite{Weinberg2} of electromagnetic radiation has $s=1$. The antighost field of this model has scaling dimension $s=2$, even though it is a scalar field and creates purely one-particle states from the vacuum. This shows that it is essential for the bounds established by the above theorem that the one-particle component of the states that the field creates from the vacuum is not of ``zero norm.''
\begin{cor*}
Let $\f$ be a field in a local covariant quantum field theory which is also scale invariant, with a nonsingular representation of dilations and Poincar\'e transformations. Assume that the possibly indefinite metric operator commutes with translations. 
\begin{itemize}
\item If $\f$ interpolates between the vacuum and physical one-particle states (one-particle states of ``nonzero norm''), then $\Box\f=0$. 
\item If $\f$ transforms in the $(j_1,j_2)$ Lorentz representation and its scaling dimension is smaller than or equal to $j_1+j_2+1$, then it does not interpolate between states of continuous mass spectrum and the vacuum, and $\Box\f=0$.   
\end{itemize}
\end{cor*}
\begin{proof} By the theorem, $\f$ does not interpolate between states of continuous mass spectrum and the vacuum. Using the expression given by Theorem \ref{matrix_integral} for $\f(f)\vac$, we can see that $\f(\Box f)\vac=0$, where $\vac$ is the vacuum. Apply Proposition \ref{reehcor}. 
\end{proof}
\section{DISCUSSION}
\label{Discussion}
``What is the simplest [interacting] quantum field theory?'' The answer proposed in [\citeonline{ArkaniHamed}] is $\mathcal{N}=4$ supersymmetric Yang-Mills theory, at least among nongravitational theories. In this section we argue that the theory may also exhibit unexpected simplicity as a local quantum field theory. We shall focus on scale invariant nonconfining Yang-Mills theories with particles charged under the global transformations associated with the gauge symmetry. $\mathcal{N}=4$ SYM theory belongs to the class of such models. Furthermore, we assume that the metric operator commutes with translations.

\subsection*{{\it One-particle states}} 

The scattering theories\cite{Lehmann,Haag,Ruelle,BuchholzFerm,BuchholzBos} developed for field theories in a positive definite inner product are based on the existence of local operators $A$ which interpolate between some one-particle state $\vf$ and the vacuum: $\langle\vf,A\vac\rangle\neq 0$. Such operators will be called interpolating operators. We also say that $\vf$ is connected to the vacuum by $A$. For the one-particle states of the models we are interested in, there are no interpolating local operators on the physical Hilbert space, so the scattering theories based on such operators are inapplicable to the theories we focus on. One of the main motivations for the BRST quantization is to remove the obstruction to local charged fields. In this formalism, there can be local interpolating operators for physical particles (one-particle states of ``nonzero norm''), which may form the basis of a scattering theory. However, scattering amplitudes are infrared divergent in nonconfining gauge theories, so these models do not have an $S$-matrix, and even the existence of one-particle states becomes questionable. For example, on the physical Hilbert space of quantum electrodynamics, there are actually no one-particle states corresponding to an electron or a positron.\cite{Schroer,Frohlich,BuchholzQED} Infrared divergences turn them into infraparticles.\cite{Schroer} Infraparticles are states of continuous mass spectrum, and a definite mass can be assigned to the electron only by the properties of the asymptotic electromagnetic field.\cite{Frohlich,BuchholzQED} 

Photons are neutral, so one can construct their asymptotic fields by Buchholz's methods\cite{BuchholzFerm,BuchholzBos} in the vacuum sector. One may consider appropriate representations of the algebra of local observables to incorporate charges. The asymptotic electromagnetic field can be defined by Haag-Ruelle type limits of some local operators in these representations as well.\cite{BuchholzQED} The momentum distribution of the incoming and outgoing electrons and positrons is fully encoded in the asymptotic electromagnetic field.\cite{Frohlich} In nonconfining Yang-Mills theories, this is not an option, since the basic fields are charged, and there is no field that could play the role of the electromagnetic field in probing the properties of charged particles. Because of this difficulty, it seems to be a reasonable assumption that one-particle states in a nonconfining Yang-Mills theory are singled out by the Poincar\'e representation, as in the vacuum sector of an ordinary quantum field theory in a positive definite metric.     

\subsection*{{\it Scaling dimensions and field equations}}

So let us assume that the basic physical fields in a nonconfining Yang-Mills theory interpolate between physical one-particle states and the vacuum. Then Theorem \ref{main_theorem} puts an upper bound on their scaling dimensions. More specifically, for scalar, spinor, and vector fields, which transform in this order in the $(0,0)$, $\big(\frac{1}{2},0\big)$ or $\big(0,\frac{1}{2}\big)$, and $\big(\frac{1}{2},\frac{1}{2}\big)$ Lorentz representation, the scaling dimension $s$ satisfies $s\leq1$, $s\leq\frac{3}{2}$, and $s\leq2$. For an antisymmetric tensor field, which transforms in $(1,0)$ or $(0,1)$, we obtain $s\leq2$. Canonical dimensions, which are determined by the scale invariance of the classical action, satisfy these inequalities.

The importance of the scaling dimensions obtained in the previous paragraph lies in the Corollary in Section \ref{Final}, which says that such fields create purely one-particle states from the vacuum, and satisfy the wave equation. In particular, the basic physical (not ghost) fields of $\mathcal{N}=4$ SYM theory have this property, and therefore their commutators are $c$-numbers by the Greenberg-Robinson theorem (see Section \ref{Reeh-schlieder}). This is quite unexpected, and it is tempting to take it as an indication that $\mathcal{N}=4$ SYM theory cannot be quantized by the BRST method, at least not under the assumptions we made. However, this possibility cannot be ruled out by the simple tests that we can perform now without delving into the details of specific constructions. 

The assumption on the interpolating property of fields might look strong, and its justification was speculative at best. Actually, the properties of BPS operators of $\mathcal{N}=4$ SYM theory add credence to this hypothesis. Canonical scaling dimensions of fields occurring in the Lagrangian have already been defined. We say that a composite field, which is defined through the product of these fields, has canonical scaling dimension if its scaling dimension is the sum of the canonical scaling dimensions of the fields in the product. The scaling dimensions of the BPS operators are known to be canonical. This implies that the basic fields also have canonical dimensions, barring the implausible scenario in which their dimensions are not canonical, but the scaling dimension of the operator product that defines the composite BPS operators receives an anomalous contribution that turns them into the canonical value.              

All the observable local operators of $\mathcal{N}=4$ SYM theory are composite, so if their anomalous dimensions vanish, they create purely states of continuous mass spectrum from the vacuum if the basic fields have canonical scaling dimensions. To see this, consider observables constructed out of the field strength, the fermion and scalar fields, and the covariant derivatives of these fields. (We assume that the covariant derivative of a dilation covariant field is also covariant under dilations, with the same scaling dimension as that of the ordinary derivative.) Take the example of an observable that arises from the product of two irreducible components of such fields, which transform in $(j_1,k_1)$ and $(j_2,k_2)$. The scaling dimensions of the two components are at least $j_1+k_1+1$ and $j_2+k_2+1$. Under the assumption of vanishing anomalous dimensions, the scaling dimension of their product is not smaller than $j_1+j_2+k_1+k_2+2$. On the other hand, all these observables transform in $(j,k)$, where $j\leq j_1+j_2$ and $k\leq k_1+k_2$, or in the direct sum of such representations. Therefore, if they create states with physical one-particle components, Theorem \ref{main_theorem} puts the upper bound $j_1+j_2+k_1+k_2+1$ on their scaling dimension, which is a contradiction. For theories in which all the one-particle states are charged, we have reached the same conclusion from the neutrality of states in the vacuum sector. The Corollary does not say anything about the equation satisfied by the observable fields.

\section*{{\it Scattering processes}}
In the scattering theory of massive particles in ordinary field theories,\cite{Haag,Ruelle} some almost local operators $A$ are taken so that they create purely one-particle states from the vacuum, and the following formula assigns time dependent operators $A^t$ to them:  
\eq{\label{AtDef}
A^t\coloneqq\midint_{\!\!x^0=t\,}\d^3x\,A(x)\overleftrightarrow{\partial}_{\!\!\!0\,}D_m(x),}
where $A(x)=\un(x)A\un(x)^*$, and $D_m$ is the Pauli-Jordan commutation function, which is defined as the solution to the Klein-Gordon equation of mass $m$ that satisfies the conditions $D_m(0,\bs{x})=0$ and $\partial_0D_m(0,\bs{x})=\delta^{(3)}(\bs{x})$. (See [\citeonline{Bogolubov}] for a textbook discussion.) In the theory of massless particles,\cite{BuchholzFerm,BuchholzBos} the operators $A$ are chosen to be local, so they do not create purely one-particle states from the vacuum $\vac$. They only interpolate between one-particle states and $\vac$. The construction of $A^t$ is more complicated than Eq.~\eqref{AtDef}, but the details do not concern us here. 

Recall that the Hilbert space $\hil$ admits an orthogonal decomposition into $\mathbb{C}\vac$, a one-particle subspace, and a subspace of continuous mass spectrum. Let $\mathcal{P}_1$ be the orthogonal projection to the one-particle subspace. Since the BRST charge $\mathrm{Q}$ is translation invariant, $\mathcal{P}_1\hil'\subset\hil'$, where $\hil'=\mathop{\mathrm{ker}}\mathrm{Q}$ is the physical subspace. Let $\vf_1,\ldots,\vf_n$ be the collection of one-particle states in $\mathcal{P}_1\hil'$. Choose a set of operators $A_1,\ldots,A_n$ such that $\vf_1=\mathcal{P}_1A_1\vac+\chi_1,\dots,\vf_n=\mathcal{P}_1A_n\vac+\chi_n$, where $\chi_1,\dots,\chi_n\in\hil''=\overline{\mathop{\mathrm{im}}\mathrm{Q}}$. The assumption is that the scattering theory of massless particles developed for ordinary field theories carries over to models in an indefinite metric. In more detail, this means first of all that the limits
\eq{\label{PsiLimits}
\Uppsi_{\mathrm{in}}(A_1,\ldots,A_n)\coloneqq\lim_{t\to-\infty}A^t{}_{\!\!\!1}\ldots A^t{}_{\!\!n\,}\vac,\;\;\;\;\;\;\;\;\Uppsi_{\mathrm{out}}(A_1,\ldots,A_n)\coloneqq\lim_{t\to\infty}A^t{}_{\!\!\!1}\ldots A^t{}_{\!\!n\,}\vac}
exist for any choice of $A_1,\dots,A_n$ from an appropriate class of operators, and they are in $\hil'$. We will not specify in what sense the limits should exist. A minimum requirement is that the limit of $(\vf,A^t{}_{\!\!\!1}\ldots A^t{}_{\!\!n\,}\vac)$ exists for any $\vf\in\hil'$. These scalar products determine $\Uppsi_{\mathrm{in/out}}(A_1,\ldots,A_n)$ only up to a vector in $\hil''$, which is enough for the purpose of describing scattering processes of physical particles. The states $\Uppsi_{\mathrm{in/out}}(A_1,\ldots,A_n)$ may depend on the choice of  $A_1,\dots,A_n$. The $n$-particle asymptotic states are linear combinations of $\Uppsi_{\mathrm{in/out}}(A_1,\ldots,A_n)$ and the states $\Uppsi_{\mathrm{in/out}}$ assigned to the same set of operators, with one or more pairs of them omitted. With fixed $\vf_1,\ldots,\vf_n$, these states should differ at most by vectors in $\hil''$ for different choices of $A_1,\dots,A_n$. 

In an ordinary field theory, the operators $A$ can be some Wightman fields smeared with compactly supported functions. In the BRST formalism, the basic fields may not be eligible because their local operators are not gauge invariant in the sense of Eq.~\eqref{GaugeInvBRST}, so they may not have the properties discussed in the previous paragraph. The reconstruction theorem\cite{Yngvason} applicable to an indefinite metric and the results listed in Section \ref{Reeh-schlieder} show that a scale invariant hermitian scalar field $\f$ that satisfies the wave equation is isomorphic to a constant multiple of the Fock free field on the subspace of states that the polynomials of $\f$ create from the vacuum. In an ordinary field theory, this would imply that particles connected to the vacuum by $\f$ do not scatter on each other. This does not follow for the physical particles connected to the vacuum by a basic scalar field $\f$ in the BRST quantization of $\mathcal{N}=4$ SYM theory even though we concluded that $\f$ satisfies the wave equation. The reason is that it may not be possible to substitute $\f(f)$ for $A$ in the previous paragraph since $\f(f)$ do not satisfy Eq.~\eqref{GaugeInvBRST}. As discussed in the Introduction, an operator $A$ that creates a charged element of $\hil'$ from the vacuum is not local relative to $\bs{E}^a$, so $A$ and $A^t$ may have localization properties insufficient for the limits in Eq.~\eqref{PsiLimits} to exist. This may be the origin of infrared divergences of the scattering amplitudes.  

\subsection*{\it{Axioms and assumptions}}

Since the BRST formalism is meant to be only a mathematical structure preliminary to the construction of physical observables, some of its axioms may be weakened as long as they still provide a framework in which a local quantum field theory can ultimately be developed on the physical Hilbert space. For instance, how important is it that the fields are distributions, in particular, that they can be evaluated on any test functions? The assumption that the fields can be smeared with any test function is essential for the analytic properties of their vacuum expectation values, without which most of the foundational theorems in field theory could not be proved. So if we relaxed it, the appeal of the BRST method would be lost. However, such a weakening of the axioms would not make the formulation inadequate, provided that the gauge invariant local fields can still be constructed even if the basic fields are defined only for a restricted class of test functions.

Consider the product $\f(f)\Uppsi(g)$ of two operators of the fields $\f$ and $\Uppsi$. Formally, this is written as the unsmeared $\f(x)\Uppsi(y)$. If this is a series of operators of increasing mass dimension, in which there is a term $C(x,y)O(y)$ such that the dimension of $O$ is the sum of the dimensions of $\f$ and $\Uppsi$, then $O(x)$ may be identified with a constant multiple of the renormalized product $\big[\f(x)\Uppsi(x)\big]_\mathrm{renorm}$, and $\big[\f(x)\Uppsi(x)\big]_\mathrm{renorm}$ is said to have vanishing anomalous dimension. More precisely, this term will appear as a constant times $\big[\smallint\d xf(x)\big] O(g)$ in the expansion of $\f(f)\Uppsi(g)$ in the limit in which the support of $f$ is shrinking to the origin. To see this, note that $C(x,y)$ behaves as $\mathcal{O}(1)$ in the limit $x\to y$ for dimensional reasons. (In order to avoid light cone singularities, the test functions $f$ and $g$ have to be chosen with some care, for example with spacelike separated supports.) The term $\big[\smallint\d xf(x)\big] O(g)$ in the expansion $\f(f)\Uppsi(g)$ can be isolated if $\smallint\d xf(x)$ is not zero, which means that the Fourier transform of $f$ does not vanish at the origin. Recall that we argued that BPS operators must be composite operators of vanishing anomalous dimension, so their constituent fields need to be smearable with such test functions. The proof of Theorem \ref{main_theorem} relies on the condition that this is possible. So we cannot get around the conclusions of this theorem by weakening the distribution property of the fields, because the basic fields could no longer serve the purpose of constructing all the physical observables.

How restrictive is the assumption that the metric operator commutes with translations? One of the simplest examples for pseudo-unitary Poincar\'e representations that do not satisfy this condition is the following. Let $\eta_0$ be a metric on $\mathbb{C}^6$ of signature $(-,-,+,+,+,+)$, and $u\in\mathbb{R}^6$ such that $u\neq 0$ and $\langle u,\eta_0u\rangle=0$, where $\langle\cdot\,,\cdot\rangle$ is the standard scalar product in $\mathbb{C}^6$. The pseudo-orthogonal group that preserves $\eta_0$ is $S\!O(2,4)$, and the stabilizer subgroup of $u$ is the Poincar\'e group. This way we obtain a six dimensional representation $S$ of the Poincar\'e group, and $\big(\un(a,L)\vf\big)(p)=e^{\im a\cdot p}\,S(a,L)\,\vf(L^{-1}p)$ is a pseudo-unitary representation on $L^2(V_m,\mathbb{C}^6,\mu_1)$ equipped with metric $(\eta\vf)(p)=\eta_0\vf(p)$, where $V_m$ is the forward mass shell of mass $m$. Here $\mu_1$ is the usual Lorentz invariant measure on $V_m$. The translation generators of any finite dimensional representation of the Poincar\'e group are nilpotent,\cite{Jacobson,Mack} so the matrices $S(a)$ are polynomials of $a$. Therefore $\un$ satisfies the spectrum condition (Axiom \ref{spectrum}), so it is a physically acceptable representation. However, there is no scalar product such that the corresponding metric operator commutes with translations.

For the sake of argument, let us assume that for a covariant field $\f$, the form of $\f(f)\vac$ is analogous to the expression in Theorem \ref{matrix_integral}:  
\[
\big(\mathcal{P}_1\f(f)\vac\big)(p)=\smallint_M\d x\,\big[e^{\im p\cdot x}f(x)S(x)\big]\,\Gamma(p),
\]
where $\mathcal{P}_1$ is the orthogonal projection to the one-particle subspace on which the Poincar\'e representation is the one specified in the previous paragraph. If $A^t$ is defined by Eq.~\eqref{AtDef}, then $\lim_{t\to\pm\infty}\langle\vf,A^t\vac\rangle$ exist for all $\vf\in\mathcal{P}_1\hil$ only if $S$ acts trivially on $\Gamma$. Otherwise, the asymptotic states cannot be defined by the usual procedure. There have been speculations\cite{Strocchi} that confinement is due to the impossibility of assigning asymptotic states to quarks in the indefinite metric formulation. This example illustrates how this can be the result of a representation in which translations do not commute with the metric. However, it has yet to be decided whether the explanation for confinement indeed lies in this phenomenon. So even in the case of confinement, the physical relevance of such representations is unclear. In nonconfining models, there seems to be no reason why they should be considered.

\subsection*{\it{Outlook}}

We have not demonstrated that if the unobservable basic fields in the BRST quantization of a scale invariant Yang-Mills theory satisfy the wave equation, then either some particles do not collide with each other, or the observable fields obey field equations that would for example result in a factorization of their correlation functions, which is characteristic of a free field theory. More analysis needs to be done in order to decide whether such unobservable fields can constitute the basis of an interacting quantum field theory.

\section*{ACKNOWLEDGMENTS}
I am grateful to Zolt\'an Zimbor\'as for useful comments on the manuscript. I would also like to thank Emil J. Martinec for discussions. This work was supported in part by DOE grant DE-FG02-90ER-40560.

\section*{APPENDIX}
Let $X$ be a topological space with a nonnegative Borel measure $\mu$ on it. For each point $x\in X$, let there be given a Hilbert space $K_x$. Assume that the support of $\mu$ admits a partition into measurable sets $(X_n)_{n=\omega,1,2,\dots}$ such that $K_x=H_\omega=\ell^2$ if $x\in X_\omega$ and $K_x=H_n=\mathbb{C}^n$ if $x\in X_n$, where $n\in\mathbb{N}^+$. Let $\langle\cdot\,,\cdot\rangle_n$ and $\|\cdot\|_n$ be the scalar product and the norm, respectively, in $H_n$. The direct integral of $K_x$ with respect to $\mu$, which is denoted by $\smallint^\oplus_{\,X}\d\mu(x)K_x$, is the equivalence classes of all sequences $\varphi=(\varphi_n)_{n\in\mathbb{N}}$ of functions $\varphi_n:X_n\to H_n$ for which $x\mapsto\langle v_n,\varphi_n(x)\rangle_n$ are measurable on $X_n$ for any sequence of vectors $(v_n)_{n\in\mathbb{N}}$, where $v_n\in H_n$, and $\sum_n\smallint_{X_n}\d\mu(x)\,\|\varphi_n(x)\|_n^2<\infty$. The sequences $\varphi$ and $\psi$ are equivalent if for each $n$ the equality $\varphi_n=\psi_n$ holds almost everywhere. The linear combination $\lambda\varphi+\kappa\psi$ of two sequences $\varphi$ and $\psi$ is $(\lambda\varphi_n+\kappa\psi_n)_{n\in\mathbb{N}}$  with the pointwise linear combination for each member. The scalar product in $\smallint^\oplus_{\,X}\d\mu(x)K_x$ is given by $\langle\varphi,\psi\rangle\coloneqq\sum_n\smallint_{X_n}\d\mu(x)\,\langle\varphi_n(x),\psi_n(x)\rangle_n$, where, with an abuse of notation, $\varphi$ and $\psi$ denotes two equivalence classes on the left hand side, whereas $\varphi_n$ and $\psi_n$ are the members of some sequences in these classes. If $H$ be a separable Hilbert space, then the space of $H$-valued square integrable functions on $X$ is $L^2(X,H,\mu)\coloneqq\smallint^\oplus_{\,X}\d\mu(x)K_x$, where $K_x=H$ for all $x$.   

If $\vf\in\smallint^\oplus_{\,X}\d\mu(x)\mathcal{H}_x$, we use the notation $\vf(x)$ for the vector $\vf_n(x)$, where $\vf_n$ is a fixed sequence from the equivalence class $\psi$ and $x\in X_n$. Any relation involving $\vf(x)$ will be independent of the choice of the representative of the equivalence class, so it is not indicated by the notation what choice has been made.

Let $B$ be a family of operator-valued functions $B_n$. Each member is defined on the components $X_n$ of a partition of $X$ and takes its value in the set $\mathcal{B}(H_n)$ of bounded operators on $H_n$. Such a family of maps can be combined into a single function $B:X\to\cup_n\mathcal{B}(H_n)$ by $B|_{X_n}\coloneqq B_n$. We say that $B$ is weakly Borel if $x\mapsto\langle u,B_n(x)v\rangle_n$ is Borel measurable for all $n$ and $u,v\in H_n$. Given a family $u$ of functions $u_n$ with each member being a map $X_n\to H_n$, $x\mapsto u(x)$ is called weakly Borel if $x\mapsto\langle v,u(x)\rangle_n$ is Borel for all $n$ and $v\in H_n$.

The proof of the following theorem is outlined in [\citeonline{Wightman}].
\begin{prop}\label{UnitaryTranslation} Every continuous unitary representation of the translation group $\trans$ on a separable Hilbert space is unitary equivalent to a representation of the following form:
\[\big(\un(a)\varphi\big)(p)=e^{\mathrm{i}p\cdot a}\varphi(p),\]
where $\varphi\in\smallint^\oplus_{\transd}\d\mu(p)H_p$ with some nonnegative Borel measure $\mu$. A bounded operator $B$ which commutes with all $\un(a)$ is decomposable, that is, can be written as 
\[(B\varphi)(p)=B(p)\varphi(p),\]
where $B(p)$ are bounded operators and $p\mapsto B(p)$ is weakly Borel.
\end{prop}

To see how $\un$ can be put into the form described in Section \ref{Representation}, first we conclude by Proposition \ref{UnitaryTranslation} that $\un$ is unitary equivalent to a representation in which translations are represented as in Eq.~\eqref{uni_general} on $\smallint^\oplus_{\transd}\d\mu(p)H_p$. Furthermore, the metric operator is decomposable. We proceed to see what restrictions the other symmetries impose on $\mu$ and $H_p$. For this purpose, we can repeat Mackey's argument for unitary representations with very little modifications. The analysis is applicable to any separable locally compact Hausdorff group $G$. One simplification specific to the group of dilations and Lorentz transformations is that $\transd$ decomposes into a finite number of orbits, but this property can be replaced by the more general requirement that the semidirect product is regular (see Section 14 in [\citeonline{Mackey3}]). We will sketch the proof. For the details omitted here, the reader is referred to [\citeonline{Mackey1,Mackey2}], more specifically to the proof of Theorem 2 in [\citeonline{Mackey2}]. Some parts of Mackey's argument are reproduced in Section 3 and Appendix IV of [\citeonline{Wightman}], where the general theory of induced representations is applied to the three dimensional Euclidean group.

Let $E\subset\transd$ be a Borel set, and $P(E)$ is the projection defined on $\smallint^\oplus_{\transd}\d\mu(p)H_p$ as the multiplication by the characteristic function of $E$. Since $\un(g)\un(a)\un(g)^{-1}=\un(ga)$, the map $P:E\mapsto P(E)$ is a system of imprimitivity for the representation $g\mapsto U(g)$ in the sense of Section 2 of [\citeonline{Mackey2}], so it obeys the transformation rule $\un(g)P(E)\un(g)^{-1}=P(gE)$. This implies that $\mu$ is quasi-invariant, that is, if $\mu(E)=0$ for a Borel set of $\transd$, then $\mu(gE)=0$ for all $g\in G$. So the Radon-Nikodym derivative in Eq.~\eqref{uni_general} exists. As for dilations and Poincar\'e transformations, there are three orbits in $\transd$ which consists of vectors of nonnegative time component: $\sigma_0=\{\,0\,\}$, $\sigma_1=\partial V^+\setminus\{\,0\,\}$, and $\sigma_\mathfrak{c}=V^+$. The projections $P$ can be written as the sum $P(E)=P_0(E\cap\sigma_0)+P_1(E\cap\sigma_1)+P_\mathfrak{c}(E\cap\sigma_\mathfrak{c})$, where $P_i$ are systems of imprimitivity on the Borel sets of $\sigma_i$, $i=0,1,\mathfrak{c}$. The restriction of $\mu$ to $\sigma_i$ gives a quasi-invariant measure $\mu_i$ on $\sigma_i$. The integration measures in Eq.~\eqref{norms} are such measures, where we chose a specific parametrization of $\sigma_i$. Let $U_i$ be the restriction of $U$ to the orthogonal subspaces $\hil_i\coloneqq P(\sigma_i)\hil$. We proceed to describe $U_i$ in more detail. For notational simplicity, the label $i$ will be dropped in the rest of the argument. So $\sigma$ can be any of the orbits $\sigma_i$, and the same convention applies to $P$, $\hil$, $U$, and $\mu$, as well as any other object associated with them.            

The action of $G$ on the orbit $\sigma$ is ergodic [\citeonline{Mackey2}], so the set of all $P(E)$
is a uniformly $n$-dimensional Abelian ring of projections [\citeonline{Mackey1}], where $n$ can also be countably infinite. So $H_p$ are the same Hilbert space $H$ for all $p$ in $\sigma$. The scalar product on $H$ will be denoted by $\langle\,,\rangle$. Choose a quasi-invariant measure $\mu$ on $\sigma$, and define the operator $W$ on $\hil=\smallint^\oplus_\sigma\d\mu(p)H_p=L^2(\sigma,H,\mu)$ by
\[(W(g)\varphi)(p)=\sqrt{\sfr{\d\mu(g^{-1}p)}{\d\mu(p)}}\,\varphi(g^{-1}p).\] 
All nontrivial quasi-invariant Borel measures on a homogeneous space of a locally compact group are equivalent each to other. The specific choice of $\mu$ is inessential because if $\mu$ and $\nu$ are two equivalent measures, then the map
\[L^2(\sigma,H,\mu)\to L^2(\sigma,H,\nu),\;\;\;\;\varphi\mapsto\sqrt{\sfr{\d\mu}{\d\nu}}\varphi\] 
is a unitary equivalence that preserves the form $e^{\im a\cdot p}\varphi(p)$ of the representation of translations. There is a measure on $\sigma_\mathfrak{c}$ which is even invariant under dilations and Lorentz transformations, but we specified a quasi-invariant one. This is only for convenience so that the action of $\un(\lambda)$ in Eq.~\eqref{MultRule} results in the same factor of $\lambda$ on both components $\varphi_1$ and $\varphi_\mathfrak{c}$.

Since the continuous operator $\mathcal{Q}(g)=\un(g)W(g)^{-1}$ commutes with $U(a)$, Proposition \ref{UnitaryTranslation} implies that for any $\varphi\in\hil$ 
\[(\mathcal{Q}(g)\varphi)(p)=Q(p,g)\varphi(p),\]
where $Q(p,g)$ are continuous operators on $H$ such that $p\mapsto Q(p,g)$ is weakly Borel. The multiplication rule $\un(g)\un(h)=\un(gh)$ implies that for all $g$, the $\mu$-measure of the set of points $p$ for which $\langle u,[Q(p,g)Q(g^{-1}p,h)-Q(p,gh)]v\rangle_H$ does not vanish is zero for any $u,v\in H$. The sets of momenta $p$ at which this equality fails to hold may vary with $u$ and $v$. If $A$ is a continuous operator on $H$, then $A=0$ if $\langle u_n,Au_m\rangle_H=0$ for all $n$ and $m$, where $u_n$ are the elements of a complete orthogonal system, which is countable by the separability of $H$. Since the measure of countably many zero measure sets is of zero measure, the operator equation \eqref{MultRuleGen} is satisfied for almost all $p$. Equation \eqref{isom0} follows from the pseudo-unitarity of $\un$ by a similar argument.

\providecommand{\href}[2]{#2}\begingroup\raggedright\endgroup
\end{document}